\documentclass[journal,twocolumn]{IEEEtranTCOM}
\usepackage{epsfig}
\usepackage{graphicx}
\usepackage{caption}
\usepackage{subcaption}
\usepackage{balance}
\usepackage{float}
\usepackage{psfrag}
\usepackage{amsfonts,amsmath,amssymb,mathrsfs,txfonts}
\newtheorem{theorem}{Theorem}
\newtheorem{lemma}{Lemma}

\newtheorem{pro}{Proposition}
\usepackage{url}
\usepackage{cite}
\usepackage{verbatim}
\usepackage{multirow}
\usepackage{multicol}
\usepackage[font=small]{caption}
\usepackage{pifont}
\usepackage{color}
\begin{document}
\title{Robust Chance-Constrained Optimization for Power-Efficient and Secure SWIPT Systems}
\author{Tuan Anh Le, Quoc-Tuan Vien, Huan X. Nguyen, Derrick Wing Kwan Ng, \\and Robert Schober
\thanks{T. A. Le, Q.-T. Vien, and H. X. Nguyen are with the Faculty of Science and Technology, Middlesex University, London, NW4 4BT, UK. D. W. K. Ng is with the University of New South Wales, Sydney, NSW, Australia. R. Schober is with the Institute for Digital Communications, Friedrich-Alexander-University Erlangen-Nurnberg, Erlangen 91058, Germany. Email: \{t.le; q.vien; h.nguyen\}@mdx.ac.uk; w.k.ng@unsw.edu.au; robert.schober@fau.de}
\thanks{This paper has been presented in part at the IEEE Global Communications Conference (GLOBECOM), Washington DC, USA, Dec., 4-8, 2016.}}
\maketitle

\begin{abstract}
In this paper, we propose beamforming schemes to simultaneously transmit data securely to multiple information receivers (IRs) while transferring power wirelessly to multiple energy-harvesting receivers (ERs). Taking into account the imperfection of the instantaneous channel state information (CSI), we introduce a chance-constrained optimization problem to minimize the total transmit power while guaranteeing \textit{data transmission reliability}, \textit{data transmission security}, and \textit{power transfer reliability}. As the proposed optimization problem is non-convex due to the chance constraints, we propose two robust reformulations of the original problem based on safe-convex-approximation techniques. Subsequently, applying semidefinite programming relaxation (SDR), the derived robust reformulations can be effectively solved by standard convex optimization packages. We show that the adopted SDR is tight and thus the globally optimal solutions of the reformulated problems can be recovered. Simulation results confirm the superiority of the proposed methods in guaranteeing transmission security compared to a baseline scheme. Furthermore, the performance of proposed methods can closely follow that of a benchmark scheme where perfect CSI is available for resource allocation.
\end{abstract}

\IEEEpeerreviewmaketitle
\section{Introduction}
In a simultaneous wireless information and power transfer (SWIPT) system, in order to harvest meaningful amounts of energy, the
energy-harvesting receivers (ERs) must be located closer to the transmitter than  the conventional information receivers (IRs) \cite{Amin2016,Xiaoming2016,Qingshi2016,DerrickPT,SchoberSep2015,Kwanbookchapter17,Duy2017}. Being closer to the transmitter, the ERs will receive stronger radio frequency (RF) signals than the IRs. Since information conveyed via RF signals is always at risks of being overheard by eavesdroppers due to the broadcast nature of wireless channels, the information intended for the IRs has to be protected in order to prevent potential eavesdropping by the ERs. Due to its high computational complexity, conventional upper-layer cryptography may cause a high energy consumption at the receivers. Therefore, such techniques may not be suitable for protecting information in SWIPT systems \cite{Xiaoming2016,Lv2015} as SWIPT devices are usually energy limited. Instead, physical layer security \cite{Xiaoming2015,QuocTuanVTC16,Tuanglobecom16}, where fading, noise, and interference are exploited, is considered to be an effective method for providing secure information transmission in SWIPT systems \cite{Xiaoming2016,HZhang16}.

Both information transmission and power transfer are equally important in SWIPT systems. Due to current hardware limitations \cite{Amin2016,Suzhi16}, the  RF-to-direct-current-energy-conversion efficiencies at ERs are typically low. As a result, a relatively high transmit power is required to compensate for the path loss in the propagation environment and the energy loss in the power-conversion circuitry. On the other hand, high transmit powers increase the risk of information leakage to the ERs. Therefore, a power-efficient strategy satisfying the information transmission and power transfer requirements of the IRs and ERs, respectively, is required. Beamforming is known to improve the power efficiency of wireless communications \cite{Farrokh,Ami,TuanTcom2013,TuanWCM2,Zhang}, and is also a promising candidate for power-efficient and secure SWIPT \cite{DerrickPT,SchoberSep2015}.

The beamforming design is usually formulated as an optimization problem taking into account the system's quality of service (QoS) requirements\footnote{The system's QoS requirements may include for example minimum signal-to-interference-plus-noise ratios (SINRs) at the IRs, minimum received powers at the ERs, and maximum tolerable leakage SINRs at the ERs.} which are specified by the system operator/designer. In particular, the channel state information (CSI) of the channels between the transmitter and the receivers, which can be obtained by appropriate channel estimation techniques in practice, is exploited to optimally control the powers and phases of the beamformer \cite{TuanTcom2014}. Hence, the CSI plays an important role in optimizing the performance of SWIPT systems. In this context, most related works have assumed the availability of perfect CSI, i.e., there are no CSI estimation errors, for the design of the beamformers for SWIPT systems, e.g., \cite{Liu2014,Jie,TuanCL2015,WeiWu2015,HZhang16,Shiyang2016}. Unfortunately, CSI estimation errors are unavoidable due to the nature of wireless channels\cite{SchoberSep2015,WangNov2015}. Achieving near-perfect CSI estimation, especially for multiple users and multiple antennas, entails a high cost in terms of the required signalling overhead. Hence, assuming perfect CSI for beamforming design either imposes a high burden on communication systems or results in a resource allocation mismatch such that the QoSs of the users of the system cannot be guaranteed. Therefore, robust beamforming designs taking into account the imperfection of the CSI are desirable for relieving the signalling-overhead burden while maintaining the users' QoSs in practical SWIPT systems.

In the literature, the imperfection of the estimated CSI for the channel between a transmitter and a receiver is usually modeled as an error vector with random elements. Due to the randomness and continuity of the error vector, an infinite number of constraints have to be met to guarantee the QoS. This leads to an intractable beamforming design problem. To overcome this obstacle, the norms of the error vectors are often assumed to be bounded by known values \cite{KaiKitJan2016,DerrickPT,SchoberSep2015,HZhang16}. Then, using the S-procedure \cite{Boyd_convex}, the QoS constraints of the SWIPT system can be replaced by a finite number of constraints representing upper bounds on the CSI errors \cite{WangNov2015,KaiKitJan2016,WangJun2015,DerrickPT,SchoberSep2015,HZhang16}. Such a conservative design approach requires an exceedingly large amount of system resources to protect rarely occurring worst cases. Hence, less conservative approaches have recently been proposed which tolerate the violation of the QoS constraints with a certain chance or probability \cite{Feng15,Fuhui2016,zheng2016,KhandakerSep2016}.

Motivated by such probabilistic approaches, this paper focuses on the design of power-efficient transmission strategies for secure SWIPT systems employing imperfect CSI. The contributions of this paper can be summarized as follows.
\begin{itemize}
\item Taking into account the imperfection of the CSI, we propose an outage-based chance optimization problem with the objective to minimize the total transmit power subject to the following three sets of QoS constraints: i) the probability/chance that the received SINRs at the IRs are above required levels is higher than predefined targets; ii) the probability/chance that the leakage SINRs at the ERs exceed secure levels is below a threshold; iii) the probability/chance that the powers received by the ERs are above required levels is greater than a prescribed value. The aforementioned three types of constraints guarantee \textit{data transmission reliability}, \textit{data transmission security}, and \textit{power transfer reliability}, respectively.
\item Since the adopted probabilistic/chance constraints are non-convex, we employ two different mathematical tools, i.e., the S-procedure \cite{Boyd_convex} and a Bernstein-type inequality \cite{Bernstein-type}, to develop two safe approximations \cite{Ben-tal} of the original optimization problem. Using semidefinite programming (SDP) relaxation, the derived safe approximations are transformed into tractable SDPs which can be optimally solved by a standard interior-point method (IPM), e.g. the SeDuMi solver in CVX \cite{Boyd}.
\item The adopted SDP relaxations of the derived safe approximations are proved to be tight by showing that the relaxed/transformed problems always yield rank-one optimal solutions, and hence also constitute optimal solutions to the safely approximated problems.  Particularly, in our rank-one proof, we introduce a novel method to convert a second-order-cone-programming (SOCP) constraint into a linear-matrix-inequality (LMI) constraint.
\item Based on the worst-case runtime of the IPM \cite{Ben-tal_13}, the computational complexities of the reformulated SDP versions of the proposed optimization problems are characterized.
\end{itemize}

This paper differs from the related works in \cite{DerrickPT,SchoberSep2015,Feng15,Fuhui2016,zheng2016,KhandakerSep2016,HZhang16} in terms of the problem formulation and the mathematical solution as follows.

\emph{Problem formulation:}
 While this paper studies a secure SWIPT wireless system, the authors of \cite{Fuhui2016} considered a secure cognitive SWIPT system.
This paper takes into account the imperfectness of the CSI of all IRs and all ERs whereas perfect knowledge of the IRs' CSI was assumed in \cite{DerrickPT,SchoberSep2015}.
 A common target in the related literature is the minimization of the transmit power while ensuring the probabilities that the secrecy rate of each IR \cite{HZhang16,zheng2016,KhandakerSep2016} and the harvested power at each ER \cite{HZhang16,KhandakerSep2016} are above certain required levels. Thereby, the required secrecy rate constitutes one of the design parameters which can be set by the system operator/designer. The authors of \cite{HZhang16} further proposed a two-stage optimization approach to decompose the constraint on the secrecy rate into constraints on the leakage SINRs at the ERs and the SINRs at the IRs by introducing two new optimization variables. Since these variables are tuned by the algorithm in \cite{HZhang16}, the operator/designer cannot enforce/guarantee any target values for the leakage SINRs or the IR SINRs. In contrast, the problem formulation in this paper enables the operator/designer to set target values for the maximum leakage SINRs at the ERs and the minimum SINRs at the IRs. As a result, the proposed problem formulation can guarantee that the IRs' leakage information at the ERs remains below a secure level should they try to eavesdrop, e.g., below the decoding sensitivity of the ERs. In contrast, the approaches in \cite{Fuhui2016,zheng2016,KhandakerSep2016,HZhang16} cannot accomplish this. If proper secrecy codes are used, e.g., \cite{Sabbaghian11}, then having only one constraint on the secrecy rate should be sufficient for a secure transmission. However, if conventional error correcting codes are used, e.g., \cite{Metzner}, it may indeed be beneficial to consider the constraints on the IR SINRs and the ER SINRs separately. In fact with conventional error correcting codes, even if the secrecy rate of the IR is kept above a certain required level \cite{Fuhui2016,zheng2016,KhandakerSep2016,HZhang16}, the ERs may still be able to eavesdrop the IR's message if their decoding sensitivity levels are lower than the leakage SINR. Hence, in this case, the schemes in \cite{Fuhui2016,zheng2016,KhandakerSep2016,HZhang16} are less secure than the proposed schemes. Also, the optimization problem considered in this paper is more challenging than its counterparts in \cite{Feng15,Fuhui2016,zheng2016,KhandakerSep2016} as secure information transmission to multiple IRs is considered whereas \cite{Feng15} does not consider secrecy at all and \cite{Fuhui2016,zheng2016,KhandakerSep2016} protect only a single IR.

\emph{Mathematical aspects:} References \cite{HZhang16}, \cite{Feng15}, and \cite{Fuhui2016} do not prove the rank-one property of the optimal beamforming solution when SDR is applied. On the other hand, the authors of \cite{zheng2016,KhandakerSep2016} have used some inequalities to transform their SOCP constraints into LMI constraints. Although the resulting transformed optimization problems are shown to yield rank-one solutions, the employed transformations reduce the size of the feasible region of the original problem which may lead to infeasibility. In contrast, in this paper, we transform an SOCP constraint into an LMI constraint without imposing any restriction. Therefore, the SDP relaxations of the considered optimization problems, i.e., the derived safe approximations, are tight.

\emph{\textbf{Notation}:}
Lower and upper case letter $y$ and $Y$: a scalar; bold lower case letter $\mathbf{y}$: a column vector; bold upper case letter $\mathbf{Y}$: a matrix; $\left\|\cdot\right\|$: the Euclidean norm; $\left\|\cdot\right\|_F$: the Frobenius norm; $(\cdot)^T$: the transpose operator; $(\cdot)^H$: the complex conjugate transpose operator; $\textrm{Tr}\left(\cdot\right)$: the trace operator; $\textrm{Pr}\left(\cdot\right)$: the probability of an event; $\mathcal{O}(\cdot)$: the big-O notation; $\mathbf{Y}\succeq \mathbf{0}$: $\mathbf{Y}$ is positive semidefinite; $\mathbf{y}\succcurlyeq \mathbf{0}$: all elements of vector $\mathbf{y}$ are non-negative; $\mathbf{I}_x$: an $x \times x$ identity matrix; $\mathbf{0}_{A \times 1}$: an $A \times 1$ vector of all zero elements; $\mathbf{0}_{A \times B}$: an $A \times B$ matrix of all zero elements; $\textrm{Re}\{\cdot\}$: the real part of a complex number; $\textrm{Eig}_{\textrm{max}}\left(\mathbf{Y}\right)$: the maximum eigenvalue of $\mathbf{Y}$; $s^+(\mathbf{Y}): \ \textrm{max} \{\textrm{Eig}_{\textrm{max}}(\mathbf{Y}),0\}$; $\textrm{vec}\left(\mathbf{Y}\right)$: stacking all the entries of $\mathbf{Y}$ into a column vector; $\mathbb{R}$: the set of all real scalars ; $\mathbb{C}^{M\times 1}$: the set of all $M\times 1$ vectors with complex elements; $\mathbb{H}^{M\times M}$: the set of all $M\times M$ Hermitian matrices; $y\sim\mathcal{CN}(0,\sigma^2)$: $y$ is a zero-mean circularly symmetric complex Gaussian random variable with variance $\sigma^2$; $\mathbf{y}\sim\mathcal{CN}(\mathbf{0},\mathbf{Y})$: $\mathbf{y}$ is a zero-mean circularly symmetric complex Gaussian random vector with covariance matrix $\mathbf{Y}$; $\mathbf{Y}^{1/2}$: the square root of $\mathbf{Y}$; rank($\mathbf{Y}$): rank of $\mathbf{Y}$; $\forall$: for all.

\section{System Model}
\begin{figure}
\centering
    \includegraphics[width=.45\textwidth]{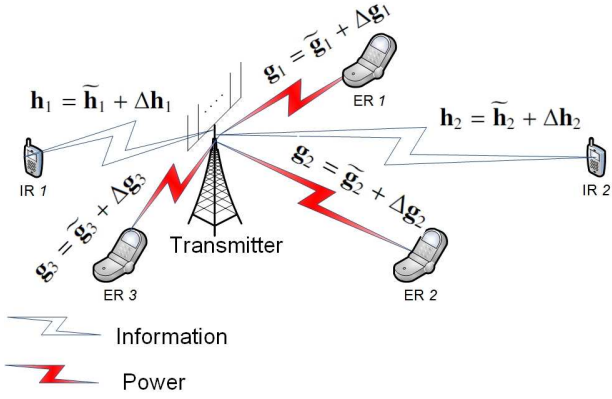}
\caption{Example of considered system model comprising a multiple-antenna transmitter, $U=2$ single-antenna IRs, and $N=3$ single-antenna ERs.} \vspace{-0.25 in}\label{sys}
\end{figure}
In this paper, we consider a downlink SWIPT system where a transmitter equipped with $M>1$ antennas simultaneously transmits information and power to $U$ IRs and $N$ ERs, respectively, using RF signals, see Fig.~\ref{sys}. Each IR and each ER is equipped with a single antenna. In order to provide secure communication, the ERs are treated as potential eavesdroppers. Hence, the power of the information-carrying signals intended for the IRs but received at the ERs should be kept low. On the other hand, due to the low energy conversion efficiency at the ERs, a high received power level is required at each ER to compensate for the power-conversion loss. Therefore, in this paper, we advocate the use of artificial noise to achieve both of these goals  \cite{DerrickPT,KaiKitJan2016,SchoberSep2015}.

Let $\mathbf{h}_{i}\in\mathbb{C}^{M\times 1}$, $i\in \{1,\cdots,U\}$, and $\mathbf{g}_{t}\in\mathbb{C}^{M \times 1}$, $t\in \{1,\cdots,N\}$, represent the actual channel coefficients of
the $i$-th IR and the $t$-th ER, respectively. Let $\mathbf{w}_{i}\in\mathbb{C}^{M \times 1}$ and $s^{(I)}_{i}\sim\mathcal{CN}(0,1)$, respectively, denote
the beamforming vector and the data for the $i$-th IR. Let  $\mathbf{v}_{t}\in\mathbb{C}^{M \times 1}$ and $s^{(E)}_{t}\sim\mathcal{CN}(0,1)$, respectively, be the artificial-noise beamforming vector and the artificial noise for the $t$-th ER. The signals received by the $i$-th IR and the $t$-th ER are, respectively, given by
\begin{eqnarray}
y^{(I)}_{i}=\sum_{j=1}^{U}\mathbf{h}^H_{i}\mathbf{w}_{j}s^{(I)}_{j}+
\sum_{t=1}^{N}\mathbf{h}^H_{i}\mathbf{v}_{t}s^{(E)}_{t}+n^{(I)}_{i}\label{signal1}
\end{eqnarray}
and
\begin{eqnarray}
y^{(E)}_{t}=\sum_{j=1}^{U}\mathbf{g}^H_{t}\mathbf{w}_{j}s^{(I)}_{j}+
\sum_{p=1}^{N}\mathbf{g}^H_{t}\mathbf{v}_{p}s^{(E)}_{p}+n^{(E)}_{t}.\label{signal2}
\end{eqnarray}
Here,  $n^{(I)}_{i}\sim\mathcal{CN}(0,\sigma^2_{I,i})$ and $n^{(E)}_{t}\sim\mathcal{CN}(0,\sigma^2_{E,t})$ are the zero-mean circularly
symmetric complex additive white Gaussian noises observed at the $i$-th IR and the $t$-th ER, respectively.

We assume that the CSI estimation at the BS is imperfect. Similar to \cite{DerrickPT,WangNov2015,KaiKitJan2016,WangJun2015,SchoberSep2015,HZhang16,Feng15,Fuhui2016,zheng2016,KhandakerSep2016}, we model the channel as $\mathbf{h}_{i}=\widetilde{\mathbf{h}}_{i}+\Delta\mathbf{h}_{i}$ and
$\mathbf{g}_{t}=\widetilde{\mathbf{g}}_{t}+\Delta\mathbf{g}_{t}$,
where $\widetilde{\mathbf{h}}_{i}\in\mathbb{C}^{M \times 1}$ and $\Delta\mathbf{h}_{i}\in\mathbb{C}^{M \times 1}$ are the estimated value of $\mathbf{h}_{i}$ and the corresponding estimation error, respectively; $\widetilde{\mathbf{g}}_{t}\in\mathbb{C}^{M \times 1}$ and $\Delta\mathbf{g}_{t}\in\mathbb{C}^{M \times 1}$ are the estimated value of $\mathbf{g}_{t}$ and the corresponding estimation error, respectively. We further assume that $\Delta\mathbf{h}_{i}\sim\mathcal{CN}(\mathbf{0},\mathbf{H}_i)$ and $\Delta\mathbf{g}_{t}\sim\mathcal{CN}(\mathbf{0},\mathbf{G}_t)$, where $\mathbf{H}_i\succeq \mathbf{0}$ and $\mathbf{G}_t\succeq \mathbf{0}$ are the channel estimation error covariance matrices which are assumed to be known for beamformer design. Interested readers are referred to \cite{Besson2008} for techniques to estimate covariance matrices. Let
$\Delta\mathbf{h}_{i}=\mathbf{H}_i^{1/2}\mathbf{e}_{i}$ and
$\Delta\mathbf{g}_{t}=\mathbf{G}_t^{1/2}\mathbf{r}_{t}$,
where $\mathbf{e}_{i}\sim\mathcal{CN}(\mathbf{0},\mathbf{I}_M)$, $\mathbf{r}_{t}\sim\mathcal{CN}(\mathbf{0},\mathbf{I}_M)$. We denote  $\{\mathbf{w}_{i}\}=\{\mathbf{w}_{1},\cdots,\mathbf{w}_{U}\}$ as the set of data beamforming vectors for all IRs and  $\{\mathbf{v}_{t}\}=\{\mathbf{v}_{1},\cdots,\mathbf{v}_{N}\}$ as the set of artificial-noise beamforming vectors.
The received SINR at the $i$-th IR, denoted by $\Gamma_{i}\left( \{\mathbf{w}_{i}\},\{\mathbf{v}_{t}\}\right)$, and the \emph{leakage SINR} of the signal intended for the $i$-th IR at the $t$-th ER, denoted by $\Gamma^{(t)}_{i}\left( \{\mathbf{w}_{i}\},\{\mathbf{v}_{t}\}\right)$, are given in \eqref{second} and \eqref{second2} at the top of next page, respectively.
\begin{figure*}
\begin{eqnarray}
\Gamma_{i}\left( \{\mathbf{w}_{i}\},\{\mathbf{v}_{t}\}\right)=\frac{\mathbf{w}_{i}^H
\left(\widetilde{\mathbf{h}}_{i}+\mathbf{H}_i^{1/2}\mathbf{e}_{i}\right)
\left(\widetilde{\mathbf{h}}_{i}+\mathbf{H}_i^{1/2}\mathbf{e}_{i}\right)^H
\mathbf{w}_{i}}
{\sum_{j=1,j \neq i}^U\mathbf{w}_{j}^H\left(\widetilde{\mathbf{h}}_{i}+\mathbf{H}_i^{1/2}
\mathbf{e}_{i}\right)
\left(\widetilde{\mathbf{h}}_{i}+\mathbf{H}_i^{1/2}\mathbf{e}_{i}\right)^H
\mathbf{w}_{j}
+\sum_{t=1}^N
\mathbf{v}_{t}^H\left(\widetilde{\mathbf{h}}_{i}+\mathbf{H}_i^{1/2}\mathbf{e}_{i}
\right)\left(\widetilde{\mathbf{h}}_{i}+\mathbf{H}_i^{1/2}\mathbf{e}_{i}\right)^H
\mathbf{v}_{t}
+\sigma^2_{I,i}},\label{second}\\
\Gamma^{(t)}_{i}\left( \{\mathbf{w}_{i}\},\{\mathbf{v}_{t}\}\right)=\frac{\mathbf{w}_{i}^H
\left(\widetilde{\mathbf{g}}_{t}+\mathbf{G}_t^{1/2}\mathbf{r}_{t}\right)\left(\widetilde{\mathbf{g}}_{t}
+\mathbf{G}_t^{1/2}\mathbf{r}_{t}\right)^H\mathbf{w}_{i}}
{\sum_{j=1,j \neq i}^U\mathbf{w}_{j}^H\left(\widetilde{\mathbf{g}}_{t}+\mathbf{G}_t^{1/2}\mathbf{r}_{t}\right)\left(\widetilde{\mathbf{g}}_{t}
+\mathbf{G}_t^{1/2}\mathbf{r}_{t}\right)^H\mathbf{w}_{j}+
\sum_{p=1}^N
\mathbf{v}_{p}^H\left(\widetilde{\mathbf{g}}_{t}+\mathbf{G}_t^{1/2}\mathbf{r}_{t}\right)\left(\widetilde{\mathbf{g}}_{t}
+\mathbf{G}_t^{1/2}\mathbf{r}_{t}\right)^H\mathbf{v}_{p}
+\sigma^2_{E,t}}.\label{second2}
\end{eqnarray}
\hrulefill
\end{figure*}
The total power received by the $t$-th ER, denoted by $\Phi_{t}\left( \{\mathbf{w}_{i}\},\{\mathbf{v}_{t}\}\right)$, is given by
\begin{eqnarray}
\Phi_{t}\left( \{\mathbf{w}_{i}\},\{\mathbf{v}_{t}\}\right)&=&
 \sum_{i=1}^U \mathbf{w}^H_{i}\left(\!\widetilde{\mathbf{g}}_{t}\!+\!\mathbf{G}_t^{1/2}\mathbf{r}_{t}\!\right)\!\left(\!\widetilde{\mathbf{g}}_{t}
+\mathbf{G}_t^{1/2}\mathbf{r}_{t}\!\right)^H\mathbf{w}_{i}\nonumber\\
&{}&+ \sum_{p=1}^N \mathbf{v}^H_{p}\left(\!\widetilde{\mathbf{g}}_{t}+\mathbf{G}_t^{1/2}\mathbf{r}_{t}\!\right)\!\left(\!\widetilde{\mathbf{g}}_{t}
+\mathbf{G}_t^{1/2}\mathbf{r}_{t}\!\right)^H\mathbf{v}_{p}. \label{TolP}
\end{eqnarray}

Hereafter, unless otherwise stated, $\{i,j\} \in \{1,\cdots,U\}$,  and $\{t,p\} \in \{1,\cdots,N\}$.
\section{Proposed Robust Chance-Constrained Optimization Problem}

The communication between the transmitter and the IRs and the power transfer to  the ERs are considered to be in QoS outage if either one of the following cases occurs: (1) The SINR level at the $i$-th IR falls below a required level $\gamma_i$, $\forall i$, which is referred to as {\it SINR outage}; (2) the leakage SINR of the $i$-th IR at the $t$-th ER is above a secure level $\gamma^{(t)}_{i}$, $\forall i, \forall t$, which is referred to as {\it leakage-SINR outage}; (3) the received power at the $t$-th ER is below a required level $P_t$, $\forall t$, which is referred to as {\it power-transfer outage}.
\subsection{Problem Formulation}

Aiming to design a power-efficient beamforming scheme, we optimize the data beamforming vector set $\{\mathbf{w}_{i}\}$ and the artificial-noise beamforming vector set $\{\mathbf{v}_{t}\}$ for minimization of the total transmit power subject to probabilistic/chance constraints on SINR outages, leakage-SINR outages, and power-transfer outages. The design is formulated as the following optimization problem:
\begin{equation}
\begin{aligned}\label{prob_02new}
& \displaystyle \min_{\{\mathbf{w}_{i}\},\{\mathbf{v}_{t}\}} & &
\sum_{i=1}^U \|\mathbf{w}_{i} \|^2+
\sum_{t=1}^N \|\mathbf{v}_{t}\|^2\\
& \text{s.\ t.}\ & &\textrm{Pr}\left(\Gamma_{i} \left( \{\mathbf{w}_{i}\},\{\mathbf{v}_{t}\}\right)\geq \gamma_{i}\right) \geq 1-\rho_i, \forall i,\\
&&& \textrm{Pr}\left(\Gamma^{(t)}_{i} \left( \{\mathbf{w}_{i}\},\{\mathbf{v}_{t}\}\right)\leq \gamma^{(t)}_{i}\right) \geq 1-\rho^{(t)}_i,  \forall i, \forall t,
\\ & & & \textrm{Pr}\left(\Phi_{t}\left( \{\mathbf{w}_{i}\},\{\mathbf{v}_{t}\}\right)\geq P_{t}\right) \geq 1-\varrho_t, \forall t,
\end{aligned}
\end{equation}
where
$\rho_i\in(0,1]$, $\rho_i^{(t)}\in(0,1]$, and $\varrho_t\in(0,1]$ are the predefined maximum tolerable probabilities/chances of SINR outages, leakage-SINR outages, and power-transfer outages, respectively.
The events $\Gamma_{i} \left( \{\mathbf{w}_{i}\},\{\mathbf{v}_{t}\}\right)\geq \gamma_{i}$ and $\Gamma^{(t)}_{i} \left( \{\mathbf{w}_{i}\},\{\mathbf{v}_{t}\}\right)\leq \gamma^{(t)}_{i}$ in the first and second sets of probabilistic/chance constraints in \eqref{prob_02new} are non-convex with respect to $\{\mathbf{w}_{i}\}$ and $\{\mathbf{v}_{t}\}$.\footnote{Note that the event $\Phi_{t}\left( \{\mathbf{w}_{i}\},\{\mathbf{v}_{t}\}\right)\geq P_{t}$ in the third constraint in \eqref{prob_02new} is convex as it is a quadratic form in $\{\mathbf{w}_{i}\}$ and $\{\mathbf{v}_{t}\}$ with positive coefficients for the second-degree terms, see \eqref{TolP}.} In the sequel, we transform these events into convex forms by introducing new variables.

To this end, we define
data beamforming matrix $\mathbf{W}_{i}=\mathbf{w}_{i}\mathbf{w}_{i}^H$ and  artificial-noise beamforming matrix $\mathbf{V}_{t}=\mathbf{v}_{t}\mathbf{v}_{t}^H$
where $\mathbf{W}_{i}\succeq \mathbf{0}$, $\mathbf{V}_{t}\succeq \mathbf{0}$, $\mathbf{W}_{i}\in \mathbb{H}^{M\times M}$, $\mathbf{V}_{t}\in \mathbb{H}^{M\times M}$, and $\mathbf{W}_{i}$  and $\mathbf{V}_{t}$ are rank-one matrices.\footnote{A matrix is rank-one
if and only if it has only one linearly independent column/row.}
Using $\mathbf{x}^{H}\mathbf{y}\mathbf{y}^H
\mathbf{x}=\mathbf{y}^{H}\mathbf{x}\mathbf{x}^H
\mathbf{y}$, we rewrite the SINR event of IR $i$, $\Gamma_{i} \left( \{\mathbf{w}_{i}\},\{\mathbf{v}_{t}\}\right)\geq \gamma_{i}$, as:
\begin{equation}
\left(\widetilde{\mathbf{h}}_{i}+\mathbf{H}_i^{1/2}\mathbf{e}_{i}\right)^H\mathbf{A}_i\left(\widetilde{\mathbf{h}}_{i}
+\mathbf{H}_i^{1/2}\mathbf{e}_{i}\right)\geq \sigma^2_{I,i},\label{co_sinr}
\end{equation}
where $\mathbf{A}_i=\left(1+\frac{1}{\gamma_i}\right)\mathbf{W}_i-\mathbf{C}$ and $\mathbf{C}=\sum_{j=1}^U\mathbf{W}_j+
\sum_{t=1}^N\mathbf{V}_t$.

Further manipulations exploiting the property $\mathbf{H}_i^H=\mathbf{H}_i$ of covariance matrices lead to the following equivalent form of \eqref{co_sinr}:
\begin{eqnarray}
f_i(\mathbf{e}_i)&\triangleq&\mathbf{e}_i^H\mathbf{H}_i^{1/2}\mathbf{A}_i
\mathbf{H}_i^{1/2}\mathbf{e}_i+2 \textrm{Re} \{\mathbf{e}_i^H\mathbf{H}_i^{1/2}\mathbf{A}_i\widetilde{\mathbf{h}}_i\}\nonumber \\&{}& +
\widetilde{\mathbf{h}}_i^H\mathbf{A}_i\widetilde{\mathbf{h}}_i-\sigma^2_{I,i}\geq 0.
\label{co_sinr2}
\end{eqnarray}
Similarly, the information leakage event of IR $i$, $\Gamma^{(t)}_{i} \left( \{\mathbf{w}_{i}\},\{\mathbf{v}_{t}\}\right)\leq \gamma^{(t)}_{i}$, can be recast as
\begin{eqnarray}
k^{(t)}_i(\mathbf{r}_{t})&\triangleq&\mathbf{r}_{t}^H\mathbf{G}_t^{1/2}\mathbf{B}_i\mathbf{G}_t^{1/2}\mathbf{r}_{t}+2 \textrm{Re} \{\mathbf{r}_{t}^H\mathbf{G}_t^{1/2}\mathbf{B}_i\widetilde{\mathbf{g}}_t\}\nonumber \\&{}& +
\widetilde{\mathbf{g}}_t^H\mathbf{B}_i\widetilde{\mathbf{g}}_t+\sigma^2_{E,t}\geq 0,\label{con_leakage}
\end{eqnarray}
where $\mathbf{B}_i=\mathbf{C}-\left(1+\frac{1}{\gamma^{(t)}_i}\right)\mathbf{W}_i$.

Furthermore, the harvested power event of ER $t$, $\Phi_{t}\left( \{\mathbf{w}_{i}\},\{\mathbf{v}_{t}\}\right)\geq P_{t}$, is equivalent to:
\begin{eqnarray}
d_t(\mathbf{r}_{t})&\triangleq&\mathbf{r}_{t}^H\mathbf{G}_t^{1/2}\mathbf{C}\mathbf{G}_t^{1/2}\mathbf{r}_{t}+2 \textrm{Re} \{\mathbf{r}_{t}^H\mathbf{G}_t^{1/2}\mathbf{C}\widetilde{\mathbf{g}}_t\}\nonumber \\&{}& +
\widetilde{\mathbf{g}}_t^H\mathbf{C}\widetilde{\mathbf{g}}_t-P_t\geq 0.\label{con_pow}
\end{eqnarray}

Using \eqref{co_sinr2}, \eqref{con_leakage}, and \eqref{con_pow},
\eqref{prob_02new} can be equivalently stated as:
\begin{equation}
\begin{aligned}\label{prob_sdp}
& \displaystyle \min_{\{\mathbf{W}_{i}\},\{\mathbf{V}_t\}\in \mathbb{H}^{M\times M}} & & \textrm{Tr}\left(\sum_{i=1}^U\mathbf{W}_{i}
+ \sum_{t=1}^N\mathbf{V}_{t}\right)\\
& \text{s.\ t.}\ & &\textrm{Pr}\left( f_i(\mathbf{e}_i)\geq 0\right)\geq 1-\rho_i, \ \forall i,
\\ & & &
\textrm{Pr}\left( k^{(t)}_i(\mathbf{r}_{t})\geq 0\right)\geq 1-\rho^{(t)}_i, \ \forall i, \forall t, \\ &&&
 \textrm{Pr}\left( d_t(\mathbf{r}_{t})\geq 0 \right)\geq 1- \varrho_t, \forall t,
 \\ &&& \mathbf{W}_{i}\succeq \mathbf{0},\ \forall i, \ \mathbf{V}_{t}\succeq \mathbf{0},\ \forall t,
 \\ &&& \textrm{rank}(\mathbf{W}_{i})=1,\ \forall i, \ \textrm{rank}(\mathbf{V}_{t})=1,\ \forall t,
\end{aligned}
\end{equation}
where
$\{\mathbf{W}_{i}\}=\{\mathbf{W}_{1},\cdots,\mathbf{W}_{U}\}$ and $\{\mathbf{V}_{t}\}=\{\mathbf{V}_{1},\cdots,\mathbf{V}_{N} \}$ are two sets of beamforming matrices. Solving problem \eqref{prob_sdp} is challenging due to the fact that the probabilistic constraints neither have simple closed-forms nor admit convexity.\footnote{Although events $f_i(\mathbf{e}_i)\geq 0$, $k^{(t)}_i(\mathbf{r}_{t})\geq 0$, and $d_t(\mathbf{r}_{t})\geq 0$ are convex, the corresponding probabilistic constraints  in \eqref{prob_sdp} are not convex.} In other words, \eqref{prob_sdp} is an NP-hard problem which cannot be solved in polynomial time. To overcome this challenge, our goal is to derive convex upper bounds for the chance constraints in \eqref{prob_sdp}.
\subsection{Safe Approximations}
First, we tackle the intractable probabilistic constraints in \eqref{prob_sdp} by replacing them by computationally tractable, i.e., convex approximations \cite{Ben-tal}. These approximations result in a convex optimization problem with respect to $\{\mathbf{W}_i\}$ and $\{\mathbf{V}_t\}$, which can be regarded as a safe approximation \cite{Ben-tal,Kun-Yu2014} if every feasible solution to the approximated problem is also feasible for the original problem \eqref{prob_sdp}.\footnote{We note that a feasible solution to the original problem \eqref{prob_sdp} may be  infeasible for the approximated problem.} In other words, the optimal solution to the safe approximation problem is a feasible suboptimal solution to the original problem. Therefore, the problem based on safe approximations serves as an upper bound for the original problem  \cite{Ben-tal,Kun-Yu2014}. Recently, three safe approximation methods have been introduced in \cite{Kun-Yu2014}, namely Method I: sphere bounding; Method II: Bernstein-type inequality; and Method III: decomposition-based large deviation inequality. In Method I, the chance constraints are approximated by assuming a spherical bound on the norm of the error vectors, while in Methods II and III, large deviation inequalities for complex Gaussian quadratic forms are utilized for constructing efficiently computable convex approximations. As reported in \cite{Kun-Yu2014}, Method II generally results in the tightest approximation among the three methods. The results in \cite{Kun-Yu2014} also indicate that in terms of power efficiency and feasibility rate,\footnote{The feasibility rate is the probability that the optimization problem is feasible \cite{Kun-Yu2014}.} Method II yields best performance followed by Method I and Method III. The poor performance of Method III is due to the fact that its approximation tightness is sacrificed for improved computational efficiency. Since, in this paper, we are interested in developing power-efficient strategies, in the following, we adopt Methods I and II to derive two safe approximations for \eqref{prob_sdp}. It is noted that our optimization problem is fundamentally different from the one considered in \cite{Kun-Yu2014} where the objective is to minimize the total transmit power subject to IR rate outage constraints for a multiple-input single-output (MISO) downlink system.
\section{Proposed Safe Approximations}
In this section, we introduce two safe approximations of \eqref{prob_sdp} and perform a complexity analysis.
\subsection{S-procedure Based Method}\label{bound_sec}
Let us assume that $\mathbf{e}_{i}$ and $\mathbf{r}_{t}$ are confined to the complex spherical sets $\xi_i\triangleq \{\mathbf{e}_{i} \in \mathbb{C}^{M\times 1}\ | \ \|\mathbf{e}_{i}\|^2\leq R_i^2\}$ and $\psi_t\triangleq \{\mathbf{r}_{t} \in \mathbb{C}^{M\times 1}\ | \ \|\mathbf{r}_{t}\|^2\leq Q_t^2\}$ having $M$ dimensions and radii $R_i$ and $Q_t$, respectively.
Since the error vector $\mathbf{e}_{i}\sim\mathcal{CN}(0,\mathbf{I}_M)$ is confined to the spherical set $\xi_i$, the probabilistic/chance constraint $\textrm{Pr}\left( f_i(\mathbf{e}_i)\geq 0\right)\geq 1-\rho_i$ holds if \cite{Kun-Yu2014}
\begin{equation}
f_i(\mathbf{e}_i)\geq 0 \ \textrm{and} \ \textrm{Pr}\left( \mathbf{e}_i \in \xi_i\right)\geq 1-\rho_i.\label{ICDF}
\end{equation}
The second condition in \eqref{ICDF} is always true if the radius of the spherical set $\xi_i$ is selected such that $R_i=\sqrt{\frac{\mathfrak{T}_m\left( 1-\rho_i\right)}{2}}$ where $\mathfrak{T}_m\left( \cdot\right)$ is the inverse cumulative distribution function of a Chi-square random variable with $m=2M$ degrees of freedom. Therefore, using $\|\mathbf{e}_i\|^2 =\mathbf{e}_i^H\mathbf{I}_M\mathbf{e}_i$, the probabilistic constraint $\textrm{Pr}\left( f_i(\mathbf{e}_i)\geq 0\right)\geq 1-\rho_i$ can be safely approximated by the following two constraints:
\begin{equation}
f_i\left( \mathbf{e}_i\right) \geq 0\ \textrm{and} \ \mathbf{e}_i^H\mathbf{I}_M\mathbf{e}_i-\frac{\mathfrak{T}_m\left( 1-\rho_i\right)}{2}\leq0.
\end{equation}

Applying similar steps for $\textrm{Pr}\left( k^{(t)}_i(\mathbf{r}_{t})\geq 0\right)\geq \rho^{(t)}_i$
and $\textrm{Pr}\left( d_t(\mathbf{r}_{t})\geq 0 \right)\geq \varrho_t$, we introduce the safe approximation of problem \eqref{prob_sdp} as
\begin{equation}
\begin{aligned}\label{prob_sdp2}
& \displaystyle \min_{\{\mathbf{W}_{i}\},\{\mathbf{V}_t\}\in \mathbb{H}^{M\times M}} & & \textrm{Tr}\left(\sum_{i=1}^U\mathbf{W}_{i}
+ \sum_{t=1}^N\mathbf{V}_{t}\right)\\
& \text{s.\ t.}\ & & f_i(\mathbf{e}_i)\geq 0, \  \mathbf{e}_i^H\mathbf{I}_M\mathbf{e}_i-\frac{\mathfrak{T}_m\left(1- \rho_i\right)}{2}\leq 0,\ \forall i,
\\ & & &
 k^{(t)}_i(\mathbf{r}_{t})\geq 0, \  \mathbf{r}_{t}^H\mathbf{I}_M\mathbf{r}_{t}-\frac{\mathfrak{T}_m\left( 1-\rho_i^{(t)}\right)}{2}\leq 0,\ \forall i, \ \forall t,
 \\ &&&
  d_t(\mathbf{r}_{t})\geq 0 , \ \mathbf{r}_{t}^H\mathbf{I}_M\mathbf{r}_{t}-\frac{\mathfrak{T}_m\left( 1-\varrho_t\right)}{2}\leq 0,\ \forall t,
 \\ &&& \mathbf{W}_{i}\succeq \mathbf{0},\  \forall i, \ \mathbf{V}_{t}\succeq \mathbf{0},\ \forall t,
 \\ &&& \textrm{rank}(\mathbf{W}_{i})=1,\ \forall i, \ \textrm{rank}(\mathbf{V}_{t})=1,\ \forall t.
\end{aligned}
\end{equation}

{\it Remark 1:} The values of the radii $R_i$ and $Q_t$ of the spherical sets $\xi_i$ and $\psi_t$ are not required in \eqref{prob_sdp2} as they have been implicitly incorporated into the outages, e.g., $R_i=\sqrt{\frac{\mathfrak{T}_m\left(1- \rho_i\right)}{2}}$. In other words, these radii are determined by the maximum tolerable outage values, e.g., $\rho_i$.

The number of constraints in \eqref{prob_sdp2} is infinite\footnote{ Problem \eqref{prob_sdp2} is a semi-infinite optimization problem, i.e., an optimization problem with a finite number of variables and an infinite number of constraints.} due to the randomness and continuousness of the error vectors $\mathbf{e}_i$ and $\mathbf{r}_{t}$. To proceed, we introduce the following lemma.
\begin{lemma}
[S-procedure\cite{Boyd_convex}] Let
$m_n(\mathbf{x})=\mathbf{x}^H\mathbf{Y}_n\mathbf{x}+2\textrm{Re}
\{\mathbf{x}^H\mathbf{y}_n\}+c_n,\ n\in\{1,2\}$, where $\mathbf{Y}_n\in \mathbb{H}^{M\times M}$, $\mathbf{y}_n\in \mathbb{C}^{M\times 1}$, and $c_n\in \mathbb{R}$. If there exists an $\check{\mathbf{x}}$ such that $m_n(\check{\mathbf{x}})<0$, then $\forall \mathbf{x} \in \mathbb{C}^{M\times 1}$, the following statements are equivalent:
\begin{enumerate}
\item $m_1(\mathbf{x})\geq 0$ and $m_2(\mathbf{x})\leq 0$ are satisfied $\forall \mathbf{x} \in \mathbb{C}^{M\times 1}$.
\item There exists a $\beta \geq 0$ such that \[
\begin{bmatrix} \mathbf{Y}_1&\mathbf{y}_1\\\ \mathbf{y}_1^H&c_1\end{bmatrix}+\beta
\begin{bmatrix} \mathbf{Y}_2&\mathbf{y}_2\\\ \mathbf{y}_2^H& c_2\end{bmatrix}\succeq \mathbf{0}.\]
\end{enumerate}
\end{lemma}
Exploiting Lemma 1 and relaxing the rank-one constraints on $\mathbf{W}_{i}$ and $\mathbf{V}_{t}$, one can transform optimization problem \eqref{prob_sdp2} into the standard convex SDP form given in \eqref{prob_sdp3} at the top of next page
\begin{figure*}
\begin{equation}
\begin{aligned}\label{prob_sdp3}
& \displaystyle \min_{\{\mathbf{W}_{i}\},\{\mathbf{V}_t\}\in \mathbb{H}^{M \times M},\alpha_i,\lambda_i^{(t)},\beta_t} & & \textrm{Tr}\left(\sum_{i=1}^U\mathbf{W}_{i}
+ \sum_{t=1}^N\mathbf{V}_{t}\right)\\
& \text{s.\ t.}\ & & \begin{bmatrix} \mathbf{H}_i^{1/2}\mathbf{A}_i\mathbf{H}_i^{1/2}+\alpha_i\mathbf{I}_M&\mathbf{H}_i^{1/2}\mathbf{A}_i\tilde{\mathbf{h}}_i\\\ \tilde{\mathbf{h}}_i^H\mathbf{A}_i\mathbf{H}_i^{1/2}&\tilde{\mathbf{h}}_i^H\mathbf{A}_i\tilde{\mathbf{h}}_i-
\sigma^2_{I,i}-\alpha_i
\frac{\mathfrak{T}_m\left( 1-\rho_i\right)}{2}\end{bmatrix}\succeq \mathbf{0},\ \alpha_i\geq0, \ \forall i,
\\ & & &
 \begin{bmatrix} \mathbf{G}_t^{1/2}\mathbf{B}_i\mathbf{G}_t^{1/2}+\lambda_i^{(t)}\mathbf{I}_M&\mathbf{G}_t^{1/2}\mathbf{B}_i\tilde{\mathbf{g}}_t\\\ \tilde{\mathbf{g}}_t^H\mathbf{B}_i\mathbf{G}_t^{1/2}&\tilde{\mathbf{g}}_t^H\mathbf{B}_i\tilde{\mathbf{g}}_t+
 \sigma^2_{E,t}-\lambda_i^{(t)}\frac{\mathfrak{T}_m\left( 1-\rho_i^{(t)}\right)}{2}\end{bmatrix}\succeq \mathbf{0}, \ \lambda_i^{(t)}\geq 0, \ \forall i, \forall t,
\\ &&& \begin{bmatrix} \mathbf{G}_t^{1/2}\mathbf{C}\mathbf{G}_t^{1/2}+\beta_t\mathbf{I}_M&\mathbf{G}_t^{1/2}\mathbf{C}\tilde{\mathbf{g}}_t\\\ \tilde{\mathbf{g}}_t^H\mathbf{C}\mathbf{G}_t^{1/2}&\tilde{\mathbf{g}}_t^H\mathbf{C}\tilde{\mathbf{g}}_t-
P_t-\beta_t
\frac{\mathfrak{T}_m\left( 1-\varrho_t\right)}{2}\end{bmatrix}\succeq \mathbf{0},\ \beta_t\geq 0, \ \forall t,
 \\ &&& \mathbf{W}_{i}\succeq \mathbf{0},\ \forall i, \ \mathbf{V}_{t}\succeq \mathbf{0},\  \forall t,
\end{aligned}
\end{equation}
\hrulefill
\end{figure*}
where $\alpha_i,\lambda_i^{(t)}$, and $\beta_t$ are auxiliary optimization variables.

{\it Remark 2:} The LMI constraints in \eqref{prob_sdp3} are similar to those of the norm-bounded approaches in e.g., \cite{WangJun2015,KaiKitJan2016}. However, the main difference between these two approaches is that the radii, i.e., the norms of the error  vectors, are predefined values in the norm-bounded approaches whereas they can be controlled  via the maximum tolerable outage probabilities in the proposed scheme.

To arrive at \eqref{prob_sdp3}, we have relaxed the rank-one constraints on the beamforming matrices $\mathbf{W}_{i}$ and $\mathbf{V}_{t}$. In the following theorem, we will show that relaxing the rank-one constraints does not affect the optimality of the solutions.

\begin{theorem}\label{theo1}
If problem \eqref{prob_sdp3} is feasible,\footnote{If \eqref{prob_sdp3} is infeasible, then some mechanisms such as admission control and quality-of-service adjustment are required at the Medium Access Control layer to restore the feasibility of the optimization problem. However, algorithms ensuring the feasibility of the optimization problem require a cross-layer design which is beyond the scope of this paper.} then its optimal solution yields rank-one matrices $\mathbf{W}_{i}$ and $\mathbf{V}_{t}$.
\end{theorem}
\begin{proof}
Please refer to Appendix \ref{apen1}.
\end{proof}
As a consequence of Theorem \ref{theo1}, the optimal beamforming vectors $\mathbf{w}_{i}^{\star}$ and $\mathbf{v}_{t}^{\star}$ are, respectively, obtained as $\mathbf{w}_{i}^{\star}=\sqrt{\lambda^{(w)}_i}\mathbf{z}^{(w)}_i$ and $\mathbf{v}_{t}^{\star}=\sqrt{\lambda^{(v)}_t}\mathbf{z}^{(v)}_t$,
where $\lambda^{(w)}_i$ and $\lambda^{(v)}_t$ are the non-zero eigenvalues and $\mathbf{z}^{(w)}_i$ and $\mathbf{z}^{(v)}_t$ are the corresponding eigenvectors of the optimal rank-one matrices $\mathbf{W}_{i}^{\star}$ and $\mathbf{V}_{t}^{\star}$, respectively. Since \eqref{prob_sdp3} is a safe approximation of \eqref{prob_sdp}, $\mathbf{W}_{i}^{\star}$ and $\mathbf{V}_{t}^{\star}$ are suboptimal solutions to \eqref{prob_sdp}. As \eqref{prob_sdp} is equivalent to \eqref{prob_02new},  $\mathbf{w}_{i}^{\star}$ and $\mathbf{v}_{t}^{\star}$ are suboptimal solutions to the original problem \eqref{prob_02new}.
\subsection{Bernstein-type-inequality Based Method}
In the previous subsection, bounded-norm conditions have been implicitly imposed on the error vectors via outage probability constraints to develop the S-procedure based method. Imposing these constraints may degrade the system performance as the feasible region of the original optimization problem is reduced.
To overcome this problem, here, we adopt a different approach to obtain another type of robust formulation, i.e., another safe convex approximation of the original problem \eqref{prob_sdp}. This convex approximation approach is based on a large deviation inequality, i.e., a Berstein-type inequality, which bounds the probability that a sum of random variables deviates from its mean \cite{Kun-Yu2014}.  To begin, let us recall the following lemma.
\begin{lemma}[Bernstein-type inequality \cite{Bernstein-type}]\label{bernstein}
Consider the following random variable $f(\mathbf{x})=\mathbf{x}^H\mathbf{Y}\mathbf{x}+2\textrm{Re}\{\mathbf{x}^H\mathbf{u}\}$, where $\mathbf{x}\sim\mathcal{CN}(\mathbf{0},\mathbf{I}_M)$, $\mathbf{Y}\in \mathbb{H}^{M\times M}$, and $\mathbf{u}\in \mathbb{C}^{M\times 1}$. For all $\delta >0$, the following statement holds:
\[\textrm{Pr}\left(f(\mathbf{x})\geq \textrm{Tr}\left( \mathbf{Y}\right)-\sqrt{2\delta}\sqrt{\|\mathbf{Y}\|_F^2+2\|\mathbf{u}\|^2}-\delta s^+(\mathbf{Y})\right)\geq 1 - e^{-\delta}.\]
\end{lemma}

With $\delta_i=-\ln{\rho_i}$ and Lemma \ref{bernstein}, the SINR outage constraint $\textrm{Pr}\left( f_i(\mathbf{e}_i)\geq 0\right)\geq 1-\rho_i$ in \eqref{prob_sdp} can be rewritten as
\begin{eqnarray}
&{}&\textrm{Tr}\left(\mathbf{H}_i^{1/2}\mathbf{A}_i\mathbf{H}_i^{1/2}\right)-\sqrt{2\delta_i}
\sqrt{\|\mathbf{H}_i^{1/2}\mathbf{A}_i\mathbf{H}_i^{1/2}\|_F^2+2\|\mathbf{H}_i^{1/2}\mathbf{A}_i
\widetilde{\mathbf{h}}_i\|^2}  \nonumber \\ &{}&- \delta_i s^+\left(\mathbf{H}_i^{1/2}\mathbf{A}_i\mathbf{H}_i^{1/2}\right)\geq \sigma^2_{I,i}-
\widetilde{\mathbf{h}}_i^H\mathbf{A}_i\widetilde{\mathbf{h}}_i. \label{bernstein1}
\end{eqnarray}
Then, by introducing two auxiliary optimization variables $\theta_i$ and $\vartheta_i$,  \eqref{bernstein1} is further recast as the following equivalent constraint:
\begin{eqnarray}
\textrm{Tr}\left(\mathbf{H}_i^{1/2}\mathbf{A}_i\mathbf{H}_i^{1/2}\right)-\sqrt{2\delta_i}\theta_i-\delta_i \vartheta_i\geq \sigma^2_{I,i}-
\widetilde{\mathbf{h}}_i^H\mathbf{A}_i\widetilde{\mathbf{h}}_i,\label{const11}\\
\sqrt{\|\mathbf{H}_i^{1/2}\mathbf{A}_i\mathbf{H}_i^{1/2}\|_F^2+2\|\mathbf{H}_i^{1/2}\mathbf{A}_i
\widetilde{\mathbf{h}}_i\|^2} \leq \theta_i,\label{const12}\\
\vartheta_i\mathbf{I}_M+\mathbf{H}_i^{1/2}\mathbf{A}_i\mathbf{H}_i^{1/2}  \succeq \mathbf{0},\label{const13}\\
\vartheta_i\geq 0\label{const14}.
\end{eqnarray}
Note that \eqref{const12} can be equivalently written as a second-order-cone (SOC) constraint
\begin{eqnarray}
\left \Vert\begin{bmatrix}
\sqrt{2}\mathbf{H}_i^{1/2}\mathbf{A}_i
\widetilde{\mathbf{h}}_i \\
\textrm{vec}\left( \mathbf{H}_i^{1/2}\mathbf{A}_i\mathbf{H}_i^{1/2}\right)
\end{bmatrix}
\right \Vert\leq \theta_i.\label{const12new}
\end{eqnarray}
Similarly, by setting $\delta^{(t)}_i=-\ln{\rho^{(t)}_i}$, and introducing two auxiliary optimization variables $\theta_i^{(t)}$ and $\vartheta_i^{(t)}$, the leakage outage constraint of the SINR, $\textrm{Pr}\left( k^{(t)}_i(\mathbf{r}_{t})\geq 0\right)\geq 1-\rho^{(t)}_i$, in \eqref{prob_sdp} can be safely approximated by the following constraints:
\begin{equation}
\textrm{Tr}\left(\mathbf{G}_t^{1/2}\mathbf{B}_i\mathbf{G}_t^{1/2}\right)\!-\sqrt{2\delta_i^{(t)}}\theta_i^{(t)}
-\delta_i^{(t)} \vartheta_i^{(t)} \geq -\sigma^2_{E,t}-
\widetilde{\mathbf{g}}_t^H\mathbf{B}_i\widetilde{\mathbf{g}}_t ,\label{const21}
\end{equation}
\begin{eqnarray}
\left \Vert\begin{bmatrix}
\sqrt{2}\mathbf{G}_t^{1/2}\mathbf{B}_i
\widetilde{\mathbf{g}}_t \\
\textrm{vec}\left( \mathbf{G}_t^{1/2}\mathbf{B}_i\mathbf{G}_t^{1/2}\right)
\end{bmatrix}
\right \Vert\leq \theta_i^{(t)} ,\label{const22}\\
\vartheta_i^{(t)}\mathbf{I}_M+\mathbf{G}_t^{1/2}\mathbf{B}_i\mathbf{G}_t^{1/2}  \succeq \mathbf{0} ,\label{const23}\\
\vartheta_i^{(t)}\geq 0\label{const24}.
\end{eqnarray}
Using the same approach, the power-transfer outage constraint, $\textrm{Pr}\left( d_t(\mathbf{r}_{t})\geq 0 \right)\geq 1- \varrho_t$, in \eqref{prob_sdp} is safely approximated by:
\begin{eqnarray}
\textrm{Tr}\left(\mathbf{G}_t^{1/2}\mathbf{C}\mathbf{G}_t^{1/2}\right)-\sqrt{2\mu_t}a_t
-\mu_t b_t\geq P_t-
\widetilde{\mathbf{g}}_t^H\mathbf{C}\widetilde{\mathbf{g}}_t,\label{const31}\\
\left \Vert\begin{bmatrix}
\sqrt{2}\mathbf{G}_t^{1/2}\mathbf{C}
\widetilde{\mathbf{g}}_t \\
\textrm{vec}\left( \mathbf{G}_t^{1/2}\mathbf{C}\mathbf{G}_t^{1/2}\right)
\end{bmatrix}
\right \Vert\leq a_t
,\label{const32}\\
b_t\mathbf{I}_M+\mathbf{G}_t^{1/2}\mathbf{C}\mathbf{G}_t^{1/2}  \succeq \mathbf{0},\label{const33}\\
b_t\geq 0,\label{const34}
\end{eqnarray}
where $a_t$ and $b_t$ are auxiliary optimization variables and $\mu_t=-\ln{\varrho_t}$. Therefore, the relaxed rank-one-constraint problem of the safe approximation of problem \eqref{prob_sdp} can be written as

\begin{equation}
\begin{aligned}\label{prob_sdp_bernstein}
& \displaystyle \min_{\{\mathbf{W}_{i}\},\{\mathbf{V}_t\}\in \mathbb{H}^{M\times M},\theta_i,\vartheta_i,\theta^{(t)}_i,\vartheta^{(t)}_i,a_t,b_t} & & \textrm{Tr}\left(\sum_{i=1}^U\mathbf{W}_{i}
+ \sum_{t=1}^N\mathbf{V}_{t}\right)\\
& \text{s.\ t.}\ & &\eqref{const11}, \eqref{const13}, \eqref{const14}, \eqref{const12new},\ \forall i,
\\ &&&
\eqref{const21}, \eqref{const22},\eqref{const23}, \eqref{const24}, \ \forall i, \forall t, \\ &&&
 \eqref{const31}, \eqref{const32},\eqref{const33}, \eqref{const34}, \forall t,
 \\ &&& \mathbf{W}_{i}\succeq \mathbf{0},\ \forall i, \ \mathbf{V}_{t}\succeq \mathbf{0},\ \forall t.
\end{aligned}
\end{equation}

The SOC constraints in \eqref{const12new}, \eqref{const22}, and \eqref{const32} are sub-cases of an SDP constraint \cite{Ami} since any SOC constraint can be recast in an LMI form using the Schur complement \cite{Boyd_convex}. Hence, the optimization problem in \eqref{prob_sdp_bernstein} is convex. Although the rank-one constraints on the beamforming matrices $\mathbf{W}_{i}$ and $\mathbf{V}_{t}$ have been relaxed to arrive at  \eqref{prob_sdp_bernstein}, the following theorem reveals that the relaxation preserves the optimality of the solution to the non-rank-one-relaxed problem , i.e., \eqref{prob_sdp_bernstein} with rank-one constraints on $\mathbf{W}_{i}$ and $\mathbf{V}_{t}$.

\begin{theorem}\label{theo2}
For the optimal solution of \eqref{prob_sdp_bernstein}, matrices $\mathbf{W}_{i}$ and $\mathbf{V}_{t}$ are always rank-one if the problem is feasible.
\end{theorem}
\begin{proof}
Please refer to Apendix \ref{apen2}.
\end{proof}
Hence, we can obtain beamforming vectors $\mathbf{w}_{i}^{\star}$ and $\mathbf{v}_{t}^{\star}$ from $\mathbf{W}_{i}^{\star}$ and $\mathbf{V}_{t}^{\star}$, respectively, using the same technique as in Section \ref{bound_sec} as robust suboptimal solutions to the original problem \eqref{prob_02new}.
\subsection{Complexity Analysis}
Hereafter, we refer to the S-procedure based method and the Bernstein-type-inequality based method as Method I and Method II, respectively. Since Method I, i.e., problem \eqref{prob_sdp3}, and Method II, i.e., problem \eqref{prob_sdp_bernstein}, contain LMI and SOC constraints, a standard IPM \cite{Boyd_convex,Ben-tal_13} can be used to find their optimal solutions. To that end, we consider the worst-case runtime of the IPM to analyze the computational complexities of the proposed methods as follows.

{\it Definition 1:} For a given $\epsilon >0$, the set of $\{\mathbf{W}_{i}^{\epsilon}\}$ and $\{\mathbf{V}_t^{\epsilon}\}$ is called an $\epsilon$-solution to problem \eqref{prob_sdp3} or \eqref{prob_sdp_bernstein} if
\begin{equation}
\textrm{Tr}\left(\sum_{i=1}^U\mathbf{W}_{i}^{\epsilon}
+ \sum_{t=1}^N\mathbf{V}_{t}^{\epsilon}\right) \leq \displaystyle \min_{\{\mathbf{W}_{i}\},\{\mathbf{V}_t\}\in \mathbb{H}^{M \times M}} \textrm{Tr}\left(\sum_{i=1}^U\mathbf{W}_{i}
+ \sum_{t=1}^N\mathbf{V}_{t}\right)+\epsilon.
\end{equation}
It can be observed that the number of decision variables of problems \eqref{prob_sdp3} and \eqref{prob_sdp_bernstein} is on the order of $(U+N) M^2$. Let $\zeta=U+N$, $\eta=\zeta+UN$, and $n=\mathcal{O}\left( \zeta M^2\right)$. We introduce the following lemma.
\begin{lemma}
\label{comlemma}
The computational complexities to obtain $\epsilon$-solutions to problems \eqref{prob_sdp3} and \eqref{prob_sdp_bernstein} are respectively
\begin{eqnarray}
\label{complexAI}
\ln(\epsilon^{-1})\sqrt{\eta(M+1)+\zeta M} \left[\eta(M+1)^2(M+1+n)\right.\nonumber\\
\left.+\zeta M^2(M+n)+n^2\right] n
\end{eqnarray}
and
\begin{eqnarray}\label{complexAII}
\ln(\epsilon^{-1})\sqrt{(\zeta+\eta)M+4\eta} \left[\eta\left((M^2+M+1)^2+2n+2\right)\right.\nonumber\\
\left.+(\zeta+\eta)M^2(M+n)+n^2\right] n.
\end{eqnarray}
\end{lemma}
\begin{proof}
Due to space limitation, we only provide a sketch of the proof. First, Method I, i.e., problem \eqref{prob_sdp3}, has $\eta$ LMI constraints of size $(M+1)$, and $\zeta$ LMI constraints of size $M$. Second, Method II, i.e., problem \eqref{prob_sdp_bernstein}, has $2\eta$ LMI constraints of size 1, $(\zeta+\eta)$ LMI constraints of size $M$, and $\eta$ SOC constraints of dimension $(M^2+M+1)$. Based on these observations, one can follow the same steps as in \cite[Section V-A]{Kun-Yu2014} to arrive at \eqref{complexAI} and \eqref{complexAII}. Note that the terms $\ln(\epsilon^{-1})\sqrt{\eta(M+1)+\zeta M}$ and $\ln(\epsilon^{-1})\sqrt{(\zeta+\eta)M+4\eta}$ in \eqref{complexAI} and \eqref{complexAII} are the iteration complexities \cite{Kun-Yu2014} required for obtaining $\epsilon$-solutions to problems \eqref{prob_sdp3} and \eqref{prob_sdp_bernstein}, respectively, while the remaining terms represent the per-iteration computation costs \cite{Kun-Yu2014}.
\end{proof}

It can be observed from Lemma \ref{comlemma} that the computational complexity of Method I is lower than that of Method II.\footnote{Consider a simple example when $M$ is large, $U=N=\frac{M}{2}$, and $n=\zeta M^2=(U+N)M^2=M^3$. The dominating terms in the complexities of Method I and Method II are $\frac{\ln(\epsilon^{-1})}{8}\sqrt{4M+9M^2+M^3}M^{10}$ and $\frac{\ln(\epsilon^{-1})}{8}\sqrt{16M+12M^2+M^3}M^{10}$, respectively.} In the following, we analyze the performance of the proposed methods in terms of power consumption.
\section{Simulation Results}
We evaluate the performance of the two proposed methods and compare them against the probabilistic-constraint-based scheme introduced in \cite{Feng15}, which is considered as the baseline scheme, the norm-bounded approach proposed in \cite{HZhang16}, and a benchmark scheme in \cite{TuanCL2015}. For the baseline scheme, the data transmission and power transfer reliabilities are guaranteed. However, secure data transmission is not considered. The SeDuMi provided by the CVX optimization package \cite{Boyd} is employed to obtain the sets of the optimal beamforming matrices $\mathbf{W}_{i}^{\star}$ and $\mathbf{V}_{t}^{\star}$.
\subsection{Simulation Setup}
Consider a transmitter supporting two IRs and two ERs, i.e., $U=N=2$. The estimated channel vectors $\widetilde{\mathbf{h}}_{i}$ and $\widetilde{\mathbf{g}}_{t}$ are respectively modelled as:
$\widetilde{\mathbf{h}}_{i}=H_i(l^{(I)}_i){\mathbf{h}}_{i,w}$ and
$\widetilde{\mathbf{g}}_{t}=G_t(l^{(E)}_t){\mathbf{g}}_{t,w}$,
where ${\mathbf{h}}_{i,w}\sim\mathcal{CN}(\mathbf{0},\mathbf{I}_M)$; ${\mathbf{g}}_{t,w}\sim\mathcal{CN}(\mathbf{0},\mathbf{I}_M)$;
$H_i(l^{(I)}_i)=\frac{c}{4\pi f_c}\left(\frac{1}{l^{(I)}_i}\right)^{\frac{\kappa}{2}}$;
$G_t(l^{(E)}_t)=\frac{c}{4\pi f_c}\left(\frac{1}{l^{(E)}_t}\right)^{\frac{\kappa}{2}}$;
$l^{(I)}_i=100$ m, $\forall i$, and $l^{(E)}_t=9$ m, $\forall t$, are the distances from the transmitter to the IRs and ERs, respectively;  $c=3\times10^8$ $\textrm{ms}^{-1}$ is the speed of light; $f_c=900$ MHz is the carrier frequency; and $\kappa=2.7$ is the path loss exponent.  The noise power at each IR and ER is assumed to be $-70$ dBm. The error covariance matrices are given as  $\mathbf{H}_i=\varepsilon \left(H_i(l^{(I)}_i)\right)^2\mathbf{I}_M$ and $\mathbf{G}_t=\varepsilon \left(G_t(l^{(E)}_t)\right)^2\mathbf{I}_M$ where $\varepsilon=0.001$. Monte-Carlo simulations have been carried out based on 500 channel realizations. The ER decoding sensitivity level is set to $\gamma_i^{(t)}=-5$ dB $\forall i, \forall t$. The SINR outage, leakage-SINR outage, and power-transfer outage probabilities are set to $10\ \%$, i.e., $\rho_i=\rho_i^{(t)}=\varrho_t=0.1, \ \forall i, \forall t$. The required power level at the ERs is $P_t=-10$ dBm, $\forall t$.
\subsection{Performance Evaluation}

\begin{figure}
\centering
  \includegraphics[width=0.45\textwidth]{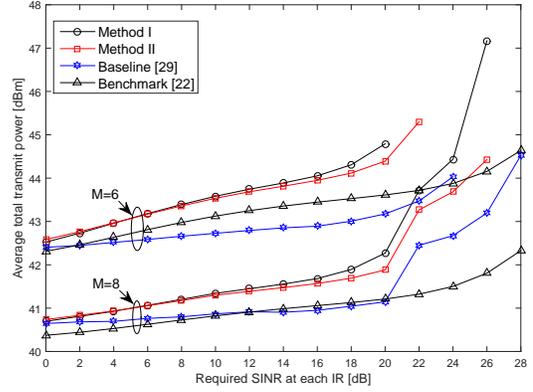}
  \caption{Average total transmit power versus required IRs' SINR for different numbers of antennas.}\vspace{-0.1 in}
  \label{Tx_SINR}
\end{figure}

\begin{figure}
  \centering
  \includegraphics[width=0.45\textwidth]{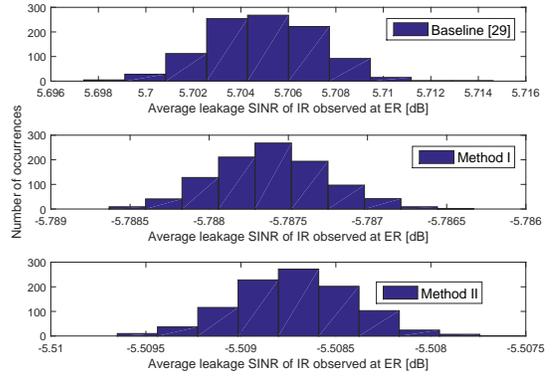}
  \caption{The histograms of the average leakage SINRs of the IR signals observed at the ERs  when the required IR SINR target is $\gamma_i=18$ dB, $\forall i$. The number of antennas is $M=6$. }\vspace{-0.1 in}
  \label{histo}
\end{figure}

Fig.~\ref{Tx_SINR} shows the average total transmit power versus the required SINR at each IR for different numbers of antennas for the proposed methods, the baseline scheme in \cite{Feng15}, and the non-robust scheme in \cite{TuanCL2015}. This figure indicates that for low-to-medium required IR SINRs, e.g. the SINR range from $0$ dB to $12$ dB, the tightenesses of the two approximations do not have much impact on their performances as the effect of the imperfect CSI is still easy to handle and hence the performances of the two proposed approaches are almost identical. However, as the IR target SINRs increase, the impact of the imperfect CSI is more severe and harder to cope with. In such situation, Method II outperforms Method I as the approximation employed in the former is tighter than that in the latter. This is due to the fact that a bounded-norm model has been implicitly imposed on the uncertainty set of the CSI for the derivation of Method I but not for Method II. The same performance trend for the two types of approximations has also been reported in \cite{Kun-Yu2014} for conventional information transmission in MISO downlink scenarios.

Fig.~\ref{Tx_SINR} reveals that the baseline scheme consumes less power, e.g. around $0.8$ dB and $0.5$ dB less at an IR target SINR of $10$ dB, than the proposed approach for 6 and 8 antennas, respectively. The price paid for this lower power consumption is that the leakage SINR cannot be controlled, i.e., secrecy cannot be guaranteed. This will be shown and discussed more in detail in Fig.~\ref{histo}. At IR target SINRs less than $6$ dB, the performances of the proposed methods are close to that of the baseline. However, as the IR target SINR increases, the performance gap between the proposed methods and the baseline widens. The reason for this is that the leakage SINR of the IRs increases with the IRs' SINR requirements. Therefore, a higher power consumption is required by the proposed approaches as more artificial noise has to be generated to impact the ERs, hence, guaranteeing the required level of security for all IRs.

Since \eqref{prob_02new} and its equivalent form \eqref{prob_sdp} are non-convex optimization problems, finding the globally optimal solution by exhaustive search is computational infeasible given the large numbers of optimization variables. Therefore, a benchmark scheme is considered in Fig.~\ref{Tx_SINR}. In particular, the benchmark scheme is based on \cite{TuanCL2015} where perfect CSI is available for resource allocation. Thus, the performance of the benchmark scheme serves as a performance upper bound for the proposed Methods I and II. It can be observed from Fig.~\ref{Tx_SINR} that in the IR target SINR range from $0$ dB to $20$ dB, the performance gap between Method II and the benchmark scheme is less than 1 dB while the gap between Method I and the benchmark scheme is around 1 dB. This underlines the accuracy of our approximations. We note that the baseline scheme consumes slightly less transmit power than the benchmark scheme does. This benefit comes at the expense of no guaranteed communication secrecy.

Fig.~\ref{histo} shows the histograms of the average leakage SINR at the ERs. To obtain the results in these figures, for each feasible channel realization,\footnote{We refer to a channel condition as feasible if all considered methods return feasible solutions.} we first generated beamforming vectors for all methods. We then generated $1000$ random error vectors for each channel realization to test the performance of each method. The resulting leakage SINRs, SINR levels, and ER power levels were averaged over the ERs and the feasible channel realizations.
The results in Fig.~\ref{histo} indicate that the information delivered to the IRs with the baseline scheme is at a very high risk of being decoded by the ERs as the leakage SINR is more than $10$ dB higher than the ER decoding sensitivity level of -$5$ dB for all occurrences. On the other hand, the proposed methods successfully guarantee secure information transmission to the IRs as they push the leakage SINR well below the ER decoding sensitivity level for all cases. The conservatism of Method I can be observed as it pushes the leakage SINR around $0.7$ dB below the required value, i.e., $0.2$ dB lower than Method II.

\begin{figure}
\centering
    \includegraphics[width=.45\textwidth]{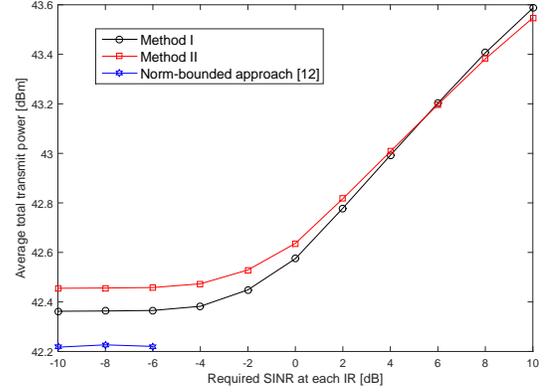}
\caption{Average total transmit power versus required IR SINR targets. The number of antennas is $M=6$. \label{NewComparisons} } \vspace{-0.1 in}
\end{figure}

In Fig. \ref{NewComparisons}, we compare our proposed methods against the norm-bounded scheme in \cite{HZhang16}. As mentioned earlier, the two-stage optimization scheme in \cite{HZhang16} can only ensure the individual secrecy rates of all IRs. In this comparison, we have given more privilege to the scheme by considering only one of those stages, i.e., \cite[Problem (36)]{HZhang16}, where we have directly set the ER decoding sensitivity level to $\gamma_i^{(t)}=-5$ dB and varied the required IR target SINR $\gamma_i$ from $-10$ dB to $10$ dB. From the figure, it is clear that the proposed methods outperform the scheme in \cite{HZhang16} in terms of power efficiency. For example, the optimization problem in \cite{HZhang16} yields infeasible solutions for all channel realizations if the required IR SINRs exceed $-6$ dB while the proposed schemes can still operate.
\begin{figure}
  \centering
  \includegraphics[width=0.45\textwidth]{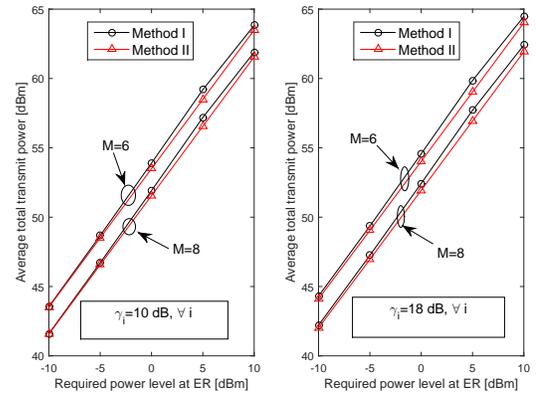}
  \caption{Average total transmit power versus required power at each ER for different required IR SINR targets.}\vspace{-0.1 in}
  \label{Tx_Erpower}
\end{figure}%
\begin{figure}
  \centering
  \includegraphics[width=0.45\textwidth]{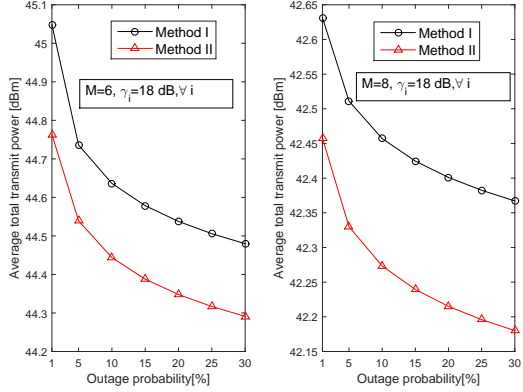}
  \caption{ Average total transmit power versus outage level for different numbers of transmit antennas.}\vspace{-0.1 in}
  \label{Tx_Erpow}
\end{figure}

Fig.~\ref{Tx_Erpower} shows the average total transmit power versus the required power level $P_t$ at the ERs for different numbers of transmit antennas and  different required IR SINRs. The average transmit power is a monotonically increasing function of the power required at the ERs. For a required receive power range from $-10$ dBm to $0$ dBm, the performances of the two proposed methods are almost identical. However, when the power required at the ER is greater than $0$ dBm, Method I consumes more power than Method II. The performance gap widens to around $0.8$ dB as the ER's power demand increases to $10$ dBm. This again confirms the improved tightness of the approximation employed in Method II compared to that employed in Method I for relatively high ER power demands. For the considered scenario, the required total transmit power is mainly determined by the required powers at the ERs which have a much more significant effect on the total transmit power than the required IR SINR, cf. Figs.~\ref{Tx_SINR} and \ref{Tx_Erpower}.

For Fig.~\ref{Tx_Erpow}, we have set identical  probability/chance values for the  maximum tolerable SINR outages, the leakage-SINR outages, and the power-transfer outages, i.e., $\rho_i=\rho_i^{(t)}=\varrho_t=\rho, \ \forall i,\ \forall t$. Hereafter, $\rho$ is referred to as outage probability. Fig.~\ref{Tx_Erpow} illustrates the average total transmit power versus the outage probability for different numbers of transmit antennas. As can be observed, a slight increase in the transmit power level significantly improves the QoS, i.e., reduces the outage probability. For instance, for $M=6$, an increase of $0.57$ dB in the transmit power of Method I and an increase of $0.47$ dB in that of Method II significantly  reduce the outage probability from $30\ \%$ to $1\ \%$.

Finally, from Figs.~\ref{Tx_SINR}, \ref{Tx_Erpower}, and \ref{Tx_Erpow}, one can conclude that increasing the number of antennas reduces the power consumption of the considered schemes. This is a result of the improved beamformer resolution due to the extra spatial degrees of freedom introduced by additional antennas.
\section{Conclusions and Future Work}
We have proposed a chance-constrained optimization problem to tackle the imperfection of the instantaneous CSI for the design of a power-efficient and secure SWIPT system. To handle non-convex QoS constraints, we have derived two robust reformulations of the proposed problem, i.e., Method I and Method II, adopting safe approximation techniques. Our analysis has revealed that Method I has a lower computational complexity than Method II. On the other hand, simulation results indicate that when the impact of imperfect CSI is more difficult to mitigate, i.e., at relatively high IR target SINRs, Method II outperforms Method I as it employs a tighter approximation. Both methods closely follow the performance of a benchmark scheme, where perfect CSI is available for resource allocation, while they outperform a baseline scheme in guaranteeing secure data transmissions. The results also show that a small increase in the transmit power leads to a significant reduction of the outage probability which in turn improves the QoS of the communication system.

A possible extension of this work is to consider a MIMO scenario where IRs and ERs are equipped with multiple-antennas. In such scenario, the expressions of the SINR, leakage SINR, and total power received at each ER will be different from those in this paper. Besides, the estimation errors will be characterized by matrices. 
As for Method I, it is possible to extend our robust beamforming design to the case of MIMO SWIPT involving error matrices. In particular, we can follow \cite[Lemma 1]{Kwan2017} and apply the S-procedure. However, as for the generalization of the proposed algorithms for Method II, the extension to the case of MIMO is a non-trivial task which requires a new problem-solving methodology and more investigation.
\section*{Acknowledgment}Tuan Anh Le, Quoc-Tuan Vien and Huan Xuan Nguyen are funded by the Newton Fund/British Council Institutional Links under Grant ID 216429427, Project code 101977. Derrick Wing Kwan Ng is supported under Australian Research Council's Discovery Early Career Researcher Award funding scheme (DE170100137). Robert  Schober is supported by the Alexander von Humboldt Professorship Program.
\appendices
\section{Proof of Theorem \ref{theo1}}\label{apen1}
Now, we investigate the structure of the optimal beamforming matrix $\mathbf{W}_i$. To this end, we rewrite \eqref{prob_sdp3} as:
\begin{equation}
\begin{aligned}\label{prob_sdp5}
& \displaystyle \min_{\{\mathbf{W}_{i}\},\{\mathbf{V}_t\}\in \mathbb{H}^{M\times M},\alpha_i,\lambda_i^{(t)},\beta_t} & & \textrm{Tr}\left(\sum_{i=1}^U\mathbf{W}_{i}
+ \sum_{t=1}^N\mathbf{V}_{t}\right)\\
& \text{s.\ t.}\ & & \mathbf{D}_i(\alpha_i)+\mathbf{E}_i^H\mathbf{A}_i\mathbf{E}_i\succeq \mathbf{0},\forall i
\\ & & &
 \mathbf{F}_i^{(t)}(\lambda_i^{(t)})+\mathbf{L}_t^H\mathbf{B}_i\mathbf{L}_t\succeq \mathbf{0},\ \forall i, \forall t,
\\ &&& \mathbf{K}_t(\beta_t)+\mathbf{L}_t^H\mathbf{C}\mathbf{L}_t\succeq \mathbf{0},\ \forall t,
 \\ &&& \alpha_i\geq0, \ \forall i,\  \lambda_i^{(t)}\geq 0, \ \forall i, \forall t, \ \beta_t\geq 0, \ \forall t, \\ &&& \mathbf{W}_{i}\succeq \mathbf{0},\ \forall i, \
   \mathbf{V}_{t}\succeq \mathbf{0},\  \forall t, \
\end{aligned}
\end{equation}
where
\begin{eqnarray}
\mathbf{D}_i(\alpha_i)&=& \begin{bmatrix} \alpha_i\mathbf{I}_M&\mathbf{0}_{M \times 1}\nonumber \\ \mathbf{0}_{1 \times M}&-\sigma^2_{I,i}-\alpha_i
\frac{\mathfrak{T}_m\left( 1-\rho_i\right)}{2}\end{bmatrix}, \nonumber\\
 \mathbf{E}_i&=& \begin{bmatrix} \mathbf{H}_i^{1/2}&\tilde{\mathbf{h}}_i\end{bmatrix}, \nonumber\\
\mathbf{F}_i^{(t)}(\lambda_i^{(t)})&=& \begin{bmatrix} \lambda_i^{(t)}\mathbf{I}_M&\mathbf{0}_{M \times 1}\nonumber \\ \mathbf{0}_{1 \times M}&-\sigma^2_{E,t}-\lambda_i^{(t)}\Omega\end{bmatrix}, \nonumber  \\
\mathbf{L}_t&=& \begin{bmatrix} \mathbf{G}_t^{1/2}&\tilde{\mathbf{g}}_t\end{bmatrix},\ \textrm{and} \nonumber\\
\mathbf{K}_t(\beta_t)&=& \begin{bmatrix} \beta_t\mathbf{I}_M&\mathbf{0}_{M \times 1}\nonumber \\ \mathbf{0}_{1 \times M}&-
P_t-\beta_t\Omega\end{bmatrix}.\nonumber
\end{eqnarray}

The Lagrangian of \eqref{prob_sdp5} can be expressed as
\begin{eqnarray}
&{}&\mathfrak{L}\left(\{\mathbf{W}_{i}\},\{\mathbf{V}_{t}\},\Upsilon \right)=\sum_{i=1}^U\textrm{Tr}\left(\mathbf{W}_{i}\right)
+ \sum_{t=1}^N\textrm{Tr}\left(\mathbf{V}_{t}\right)\nonumber \\&{}&-\sum_{i=1}^U\textrm{Tr}\left(
\mathbf{Q}_i\left[ \mathbf{D}_i(\alpha_i)+\mathbf{E}_i^H\mathbf{A}_i\mathbf{E}_i\right]\right)
-\sum_{i=1}^U\kappa_i \alpha_i\nonumber \\
&{}&-\sum_{i=1}^U\sum_{t=1}^N\textrm{Tr}\left(\mathbf{R}_i^{(t)}\left[ \mathbf{F}_i^{(t)}(\lambda_i^{(t)})+\mathbf{L}_t^H\mathbf{B}_i\mathbf{L}_t\right]\right)
\nonumber \\
&{}&-\sum_{i=1}^U\sum_{t=1}^N \mu_i^{(t)} \lambda_i^{(t)}-\sum_{t=1}^N\textrm{Tr}\left( \mathbf{S}_t\left[\mathbf{K}_t(\beta_t)+\mathbf{L}_t^H\mathbf{C}\mathbf{L}_t \right]\right)\nonumber \\
&{}&-\sum_{t=1}^N\psi_t \beta_t -\sum_{i=1}^U\textrm{Tr}\left( \mathbf{T}_i \mathbf{W}_i\right)
-\sum_{t=1}^N\textrm{Tr}\left( \mathbf{Z}_t \mathbf{V}_t\right),\label{lagran}
\end{eqnarray}
where $\mathbf{Q}_i$, $\kappa_i$, $\mathbf{R}_i^{(t)}$, $\mu_i^{(t)}$, $\mathbf{S}_t$, $\psi_t$, $\mathbf{T}_i$, and $\mathbf{Z}_t$ are the Lagrange multipliers associated with the constraints in \eqref{prob_sdp5}. Furthermore, we have
 $\boldsymbol\alpha=\begin{bmatrix} \alpha_{1},\cdots,\alpha_{U}\end{bmatrix}^T$, $\boldsymbol\mu=\begin{bmatrix} \mu_1^1,\cdots,\mu_U^1,\cdots,
\mu_1^N,\cdots,\mu_U^N,\end{bmatrix}^T$, $\boldsymbol\kappa=\begin{bmatrix} \kappa_{1},\cdots,\kappa_{U}\end{bmatrix}^T$,
$\boldsymbol\psi=\begin{bmatrix} \psi_{1},\cdots,\psi_{N}\end{bmatrix}^T$,
$\{\mathbf{Q}_{i}\}=\{\mathbf{Q}_{1},\cdots,\mathbf{Q}_{U} \}$,
$\{\mathbf{R}_{i}^{(t)}\}=\{\mathbf{R}_{1}^{(1)},\cdots,\mathbf{R}_{U}^{(1)},\cdots,
 \mathbf{R}_{1}^{(N)},\cdots,\mathbf{R}_{U}^{(N)}\}$,
$\{\mathbf{S}_t\}=\{\mathbf{S}_{1},\cdots,\mathbf{S}_{N}\}$,
$\{\mathbf{T}_{i}\}=\{\mathbf{T}_{1},\cdots,\mathbf{T}_{U} \}$,
 $\{\mathbf{Z}_{t}\}=\{\mathbf{Z}_{1},\cdots,\mathbf{Z}_{N} \}$, and finally \\ $\Upsilon=\left\{
\boldsymbol\alpha,\boldsymbol\kappa,\boldsymbol\mu,
\boldsymbol\psi,\{\mathbf{Q}_{i}\},\{\mathbf{R}_{i}^{(t)}\},
\{\mathbf{S}_{t}\},\{\mathbf{T}_{i}\},\{\mathbf{Z}_t\}\right\}
$.

Let $g\left(\Upsilon \right)$ be the dual function of \eqref{prob_sdp5} given by
\begin{eqnarray}\label{dual_func1}
g\left(\Upsilon \right)=\displaystyle \min_{\{\mathbf{W}_{i}\},\{\mathbf{V}_{t}\}\in \mathbb{H}^{M\times M}}
\mathfrak{L}\left(\{\mathbf{W}_{i}\},\{\mathbf{V}_{t}\},\Upsilon\right). \end{eqnarray}
Therefore, the corresponding dual problem of \eqref{prob_sdp5} can be stated as:
\begin{equation}\label{dual}
\begin{aligned}
&\displaystyle \max_{\Upsilon}
 &&g\left(\Upsilon \right),\\
& \text{s.\ t.}\ & &\boldsymbol\alpha\succcurlyeq \mathbf{0},\boldsymbol\kappa\succcurlyeq \mathbf{0}, \ \boldsymbol\mu\succcurlyeq \mathbf{0},\ \boldsymbol\psi\succcurlyeq \mathbf{0},\\&&& \mathbf{Q}_{i} \succeq \mathbf{0},\mathbf{T}_i \succeq \mathbf{0},\  \mathbf{S}_{t} \succeq \mathbf{0},\mathbf{Z}_t \succeq \mathbf{0}, \ \mathbf{R}_i^{(t)} \succeq \mathbf{0},\ \forall i, \ \forall t.
\end{aligned}
\end{equation}

Since problem \eqref{prob_sdp5} is convex and satisfies Slater's constraint qualification \cite{Boyd_convex}, the duality gap is zero and the optimal solution of \eqref{prob_sdp5} can be obtained by solving \eqref{dual}.

Let $\Upsilon^{\star}=\{
\boldsymbol\alpha^{\star},\boldsymbol\kappa^{\star},\boldsymbol\mu^{\star},
\boldsymbol\psi^{\star},\{\mathbf{Q}_{i}^{\star}\},\{\mathbf{R}_{i}^{(t){\star}}\},\{\mathbf{S}_{t}^{\star}\},
\{\mathbf{T}_{i}^{\star}\},\{\mathbf{Z}_t^{\star}\}\}$ be the optimal solution to dual problem \eqref{dual}, then the corresponding
optimal solution $\left(\{\mathbf{W}_{i}^{\star}\},\{\mathbf{V}_t^{\star}\}\right)$ to problem \eqref{prob_sdp5}
can be obtained as
\begin{eqnarray}\label{dual_func}
g(\Upsilon^{\star})=\displaystyle \min_{\{\mathbf{W}_{i}\},\{\mathbf{V}_{t}\}\in \mathbb{H}^{M\times M}}
\mathfrak{L}\left(\{\mathbf{W}_{i}\},\{\mathbf{V}_{t}\},\Upsilon^{\star}\right).
\end{eqnarray}

Substituting for $\mathbf{A}_{i}$, $\mathbf{B}_{i}$, and $\mathbf{C}$ in \eqref{lagran}, after some mathematical manipulations, one can arrive at
\begin{equation}
\mathfrak{L}\left(\{\mathbf{W}_{i}\},\{\mathbf{V}_{t}\},\Upsilon^{\star}\right)=\sum_{i=1}^U \textrm{Tr}\left(\mathbf{\Theta}_{i} \mathbf{W}_{i}\right)+\sum_{t=1}^N \textrm{Tr}\left(\mathbf{\Xi}_{t} \mathbf{V}_{t}\right)+\chi
\end{equation}
where $\mathbf{\Theta}_{i}=\mathbf{I}_M -\frac{1}{\gamma_i}\mathbf{E}_i\mathbf{Q}_i^{\star}\mathbf{E}_i^H+
\sum_{t=1}^N\frac{1}{\gamma_i^{(t)}}\mathbf{L}_t\mathbf{R}_i^{(t)\star}\mathbf{L}_t^H-\sum_{t=1}^N\mathbf{L}_t\mathbf{S}_t^{\star}\mathbf{L}_t^H-\mathbf{T}_i^{\star}$, $\mathbf{\Xi}_{t}=\mathbf{I}_M +\sum_{i=1}^U\mathbf{E}_i\mathbf{Q}_i^{\star}\mathbf{E}_i^H+
\sum_{p=1}^N\sum_{i=1}^U\mathbf{L}_p\mathbf{R}_i^{(p)\star}\mathbf{L}_p^H-\sum_{p=1}^N\mathbf{L}_p\mathbf{S}_p^{\star}\mathbf{L}_p^H-\mathbf{Z}_t^{\star}$,
and
\begin{eqnarray}
&{}&\chi=-\sum_{i=1}^U\textrm{Tr}\left(
\mathbf{Q}_i^{\star}\left[ \mathbf{D}_i(\alpha_i^{\star})+\mathbf{E}_i^H\left(\sum_{j=1,j\neq i}^U\mathbf{W}_j\right)\mathbf{E}_i\right]\right)\nonumber \\
&{}&-\sum_{i=1}^U\sum_{t=1}^N\textrm{Tr}\left(\mathbf{R}_i^{(t)\star}\left[ \mathbf{F}_i^{(t)}(\lambda_i^{(t)\star})+\mathbf{L}_t^H\left(\sum_{j=1,j\neq i}^U\mathbf{W}_j\right)\mathbf{L}_t\right]\right)\nonumber \\
&{}&-\sum_{i=1}^U\kappa_i^{\star} \alpha_i^{\star}-\sum_{i=1}^U\sum_{t=1}^N \mu_i^{(t)\star} \lambda_i^{(t)\star}-\sum_{t=1}^N \psi_t^{\star}\beta_t^{\star}.\label{xi}
\end{eqnarray}

Therefore, we can rewrite \eqref{dual_func} as
\begin{eqnarray}\label{dual_func2}
g(\Upsilon^{\star})
=\displaystyle \min_{\{\mathbf{W}_{i}\}\in \mathbb{H}^{M\times M}} \sum_{i=1}^U \textrm{Tr}\left(\mathbf{\Theta}_{i} \mathbf{W}_{i}\right)+
\displaystyle \min_{\{\mathbf{V}_{t}\}\in \mathbb{H}^{M\times M}} \sum_{t=1}^N \textrm{Tr}\left(\mathbf{\Xi}_{t} \mathbf{V}_{t}\right).
\end{eqnarray}

Since \eqref{prob_sdp3} is feasible, the optimal value of its equivalent form, i.e., \eqref{prob_sdp5}, is non-negative. Furthermore, the optimal duality gap between primary problem \eqref{prob_sdp5} and its Lagrange dual problem \eqref{dual} is zero.
Therefore, $\mathbf{\Theta}_{i}$ and $\mathbf{\Xi}_{t}$ must be positive semi-definite, i.e.,
$\mathbf{\Theta}_{i}\succeq \mathbf{0},\ \forall i,$ and $\mathbf{\Xi}_{t}\succeq \mathbf{0},\ \forall t$, to ensure that the Lagrangian dual function is bounded from below, i.e., the Lagrangian dual function cannot become -$\infty$. We continue by introducing the following proposition:

\begin{pro}\label{proposition1}
If an $M\times M$ Hermitian matrix $\mathbf{W}_{i}$ has a rank of $K\leq M$, then it can be expressed as
$\mathbf{W}_{i}=\sum_{j=1}^{K}\nu_{i,k}
\mathbf{a}_{i,k}\mathbf{a}_{i,k}^\text{H}$,
where $\nu_{i,k}$ and $\mathbf{a}_{i,k}$ are the $k$th non-zero eigenvalue and the
corresponding eigenvector of $\mathbf{W}_{i}$, respectively.\footnote{Proposition \ref{proposition1} can be proved using the facts that the Hermitian matrix $\mathbf{W}_{i}$ has $K$ real non-zero eigenvalues and $K$ orthogonal eigenvectors.}
\end{pro}

In the following, we prove the rank-one property of the solution of \eqref{dual_func2} by contradiction.  Assuming that the optimal solution of \eqref{dual_func2}, $\mathbf{W}_{i}^{\star}$, has rank $K>1$, $\forall i$. Proposition \ref{proposition1} indicates that  $\mathbf{W}_{i}^{\star}=\sum_{k=1}^K\nu_{i,k}
\mathbf{a}_{i,k}\mathbf{a}_{i,k}^\text{H}$. Now we construct another feasible solution to \eqref{dual_func2} as
\begin{equation}\label{subop}
\overline{\mathbf{W}}_{i}^{\star}=\nu_{i,p} \mathbf{a}_{i,p}\mathbf{a}_{i,p}^\text{H}, \ \forall i,
\end{equation}
where $p=\displaystyle \textrm{arg}\min_{k\in\{1,\cdots,K\}}\nu_{i,k}
\mathbf{a}_{i,k}^\text{H}\mathbf{\Theta}_{i}\mathbf{a}_{i,k}$. Combining $\mathbf{\Theta}_{i}\succeq \mathbf{0}$ and \eqref{subop} reveals that
\begin{equation}\label{contradict}
\sum_{i=1}^U \textrm{Tr}\left(\mathbf{\Theta}_{i} \overline{\mathbf{W}}_{i}^{\star}\right)<
\sum_{i=1}^U  \textrm{Tr}\left(\mathbf{\Theta}_{i} \mathbf{W}_{i}^{\star}\right).
 \end{equation}
The inequality in \eqref{contradict} contradicts the optimality of $\mathbf{W}_{i}^{\star}$. Therefore, $\mathbf{W}_{i}^{\star}$ must be a rank-one matrix for all $i$. Following similar arguments, one can show that $\mathbf{V}_{t}^{\star}$ must be a rank-one matrix for all $t$.
\section{Proof of Theorem \ref{theo2}}\label{apen2}
This proof is based on a fact that the SOC constraints in \eqref{const12new}, \eqref{const22}, and \eqref{const32} can be cast as LMIs. With some abuse of notation, in this proof we reuse $\mathbf{B}_i$, $\mathbf{C}$, $\mathbf{D}_i$, $\mathbf{E}_i$, $\mathbf{T}_i$, $\mathbf{P}_i$, $\mathbf{K}_p$, and $\mathbf{u}_i$ from Section III and Appendix \ref{apen1}.  In the sequel, we first show the transformation of \eqref{const12new}. 
Exploiting the Schur complement with some mathematical manipulations,  \eqref{const12new} can be equivalently written as $\mathbf{B}_i \succeq \mathbf{0}$ where
\small
\begin{eqnarray}
\mathbf{B}_i=\begin{bmatrix}\theta_i \mathbf{I}_{M^2+M} & \begin{bmatrix}
\sqrt{2}\mathbf{H}_i^{1/2}\mathbf{A}_i
\widetilde{\mathbf{h}}_i \\
\textrm{vec}\left( \mathbf{H}_i^{1/2}\mathbf{A}_i\mathbf{H}_i^{1/2}\right)
\end{bmatrix}\\
\begin{bmatrix}
\sqrt{2}\mathbf{H}_i^{1/2}\mathbf{A}_i
\widetilde{\mathbf{h}}_i \\
\textrm{vec}\left( \mathbf{H}_i^{1/2}\mathbf{A}_i\mathbf{H}_i^{1/2}\right)
\end{bmatrix}^H&\theta_i\end{bmatrix}.\label{theo2a1}
\end{eqnarray}
\normalsize
We further decompose  $\mathbf{B}_i$ as
\begin{equation}\label{P22}
\mathbf{B}_i=\mathbf{C}\left( \theta_i\right)+\mathbf{D}_i+\mathbf{D}_i^H+\mathbf{E}_i+\mathbf{E}_i^H
\end{equation}
where
\begin{eqnarray}
\mathbf{C}\left( \theta_i\right)&=&\begin{bmatrix}\theta_i \mathbf{I}_{M^2+M} & {\mathbf{0}}_{(M^2+M) \times 1} \\
{\mathbf{0}}_{1 \times (M^2+M)}^H&\theta_i\end{bmatrix},\nonumber \\
\mathbf{D}_i&=&\begin{bmatrix} \mathbf{0}_{(M^2+M)\times(M^2+M)} &
{\mathbf{0}}_{(M^2+M) \times 1}\\
\begin{bmatrix}
\sqrt{2}\mathbf{H}_i^{1/2}\mathbf{A}_i
\widetilde{\mathbf{h}}_i \\
{\mathbf{0}}_{M^2\times 1}
\end{bmatrix}^H&0\end{bmatrix},\nonumber \\
\mathbf{E}_i&=&\begin{bmatrix} \mathbf{0}_{(M^2+M)\times (M^2+M)} & {\mathbf{0}}_{(M^2+M) \times 1}\\
\begin{bmatrix}
{\mathbf{0}}_{M\times 1}\\
\textrm{vec}\left( \mathbf{H}_i^{1/2}\mathbf{A}_i\mathbf{H}_i^{1/2}\right)
\end{bmatrix}^H& 0\end{bmatrix}.\nonumber
\end{eqnarray}
Moreover, matrix $\mathbf{D}_i$ can be further decomposed as
\begin{eqnarray}
\mathbf{D}_i&=&\sqrt{2}\begin{bmatrix}
{\mathbf{0}}_{(M^2+M) \times 1}\\ 1 \end{bmatrix}\widetilde{\mathbf{h}}_i ^H
\mathbf{A}_i \mathbf{H}_i^{1/2}\begin{bmatrix} \mathbf{I}_M& {\mathbf{0}}_{M \times 1}& \cdots & {\mathbf{0}}_{M \times 1}\end{bmatrix}\nonumber\\&=&\mathbf{T}_i\left( \widetilde{\mathbf{h}}_i\right) \mathbf{A}_i \mathbf{P}_i\left( \mathbf{H}_i^{1/2}\right),\label{theo2a3}
\end{eqnarray}
where
$\mathbf{T}_i\left( \widetilde{\mathbf{h}}_i\right)=\sqrt{2}\begin{bmatrix}
{\mathbf{0}}_{(M^2+M) \times 1}\\ 1 \end{bmatrix}\widetilde{\mathbf{h}}_i ^H$,  $\mathbf{P}_i\left( \mathbf{H}_i^{1/2}\right)$ = $\mathbf{H}_i^{1/2}\begin{bmatrix} \mathbf{I}_M& {\mathbf{0}}_{M\times 1}& \cdots & {\mathbf{0}}_{M \times 1}\end{bmatrix}$,  and $\begin{bmatrix} \mathbf{I}_M& {\mathbf{0}}_{M \times 1}& \cdots & {\mathbf{0}}_{M \times 1}\end{bmatrix}$ is an $M\times(M^2+M+1)$ matrix. Let
$\mathbf{U}_i=\begin{bmatrix}
{\mathbf{0}}_{M\times 1}&\cdots &{\mathbf{0}}_{M\times 1}& \mathbf{u}_i \end{bmatrix}^T$
be an $(M^2+M+1)\times M$ matrix where $\mathbf{u}_i$ is an $M \times 1$ vector with $1$ at the $i$-th entry and all zeros elsewhere. Furthermore, let
$\mathbf{K}_i=\begin{bmatrix}
\underbrace{\mathbf{0}_{M \times M}}_{1\textrm{st}}&\mathbf{0}_{M \times M}&\cdots&\underbrace{\mathbf{I}_M}_{i\textrm{th}}&\cdots& \underbrace{\mathbf{0}_{M\times M}}_{(M^2+M)\textrm{th}}& {\mathbf{0}}_{M \times 1} \end{bmatrix}$
be an $M\times(M^2+M+1)$ matrix with $\mathbf{I}_M$ as the $i$-th block and all zeros elsewhere. One can express $\mathbf{E}_i$ as
\begin{eqnarray}
\mathbf{E}_i=\sum_{p=1}^M \mathbf{U}_p \mathbf{H}_i^{1/2}\mathbf{A}_i\mathbf{H}_i^{1/2}\mathbf{K}_p.\label{theo2a4}
\end{eqnarray}

From \eqref{theo2a1}, \eqref{P22}, \eqref{theo2a3}, and \eqref{theo2a4}, we conclude that constraint \eqref{const12new} is equivalent to the following LMI constraint
\begin{eqnarray}\small\label{const12C}
\mathbf{C}\left( \theta_i\right)+\mathbf{T}_i\left( \widetilde{\mathbf{h}}_i\right) \mathbf{A}_i \mathbf{P}_i\left( \mathbf{H}_i^{1/2}\right)+\left[\mathbf{T}_i\left( \widetilde{\mathbf{h}}_i\right) \mathbf{A}_i \mathbf{P}_i\left( \mathbf{H}_i^{1/2}\right)\right]^H\nonumber\\
+\sum_{p=1}^M \mathbf{U}_p \mathbf{H}_i^{1/2}\mathbf{A}_i\mathbf{H}_i^{1/2}\mathbf{K}_p+\left[\sum_{p=1}^M \mathbf{U}_p \mathbf{H}_i^{1/2}\mathbf{A}_i\mathbf{H}_i^{1/2}\mathbf{K}_p \right]^H \succeq \mathbf{0}.
\end{eqnarray}
Similarly, constraint \eqref{const22} can be equivalently cast in LMI form as
\begin{eqnarray}\label{const22C}
\mathbf{C}\left( \theta_i^{(t)}\right)+\mathbf{T}_i\left( \widetilde{\mathbf{g}}_t\right) \mathbf{B}_i \mathbf{P}_i\left( \mathbf{G}_t^{1/2}\right)+\left[\mathbf{T}_i\left( \widetilde{\mathbf{g}}_t\right) \mathbf{B}_i \mathbf{P}_i\left( \mathbf{G}_t^{1/2}\right)\right]^H\nonumber\\
+\sum_{p=1}^M \mathbf{U}_p \mathbf{G}_t^{1/2}\mathbf{B}_i\mathbf{G}_t^{1/2}\mathbf{K}_p+\left[\sum_{p=1}^M \mathbf{U}_p \mathbf{G}_t^{1/2}\mathbf{B}_i\mathbf{G}_t^{1/2}\mathbf{K}_p \right]^H \succeq \mathbf{0}.
\end{eqnarray}
Finally, constraint \eqref{const32} can be equivalently written as the following LMI
\begin{eqnarray}\label{const32C}
\mathbf{C}\left( a_t\right)+\mathbf{T}_i\left( \widetilde{\mathbf{g}}_t\right) \mathbf{C} \mathbf{P}_i\left( \mathbf{G}_t^{1/2}\right)+\left[\mathbf{T}_i\left( \widetilde{\mathbf{g}}_t\right) \mathbf{C} \mathbf{P}_i\left( \mathbf{G}_t^{1/2}\right)\right]^H\nonumber\\
+\sum_{p=1}^M \mathbf{U}_p \mathbf{G}_t^{1/2}\mathbf{C}\mathbf{G}_t^{1/2}\mathbf{K}_p+\left[\sum_{p=1}^M \mathbf{U}_p \mathbf{G}_t^{1/2}\mathbf{C}\mathbf{G}_t^{1/2}\mathbf{K}_p \right]^H \succeq \mathbf{0}.
\end{eqnarray}

Using \eqref{const12C}-\eqref{const32C}, one can rewrite \eqref{prob_sdp_bernstein} as
\begin{equation}
\begin{aligned}\label{prob_sdp_bernstein_Equivalent}
& \displaystyle \min_{\{\mathbf{W}_{i}\},\{\mathbf{V}_t\}\in \mathbb{H}^{M\times M},\theta_i,\vartheta_i,\theta^{(t)}_i,\vartheta^{(t)}_i,a_t,b_t} & & \textrm{Tr}\left(\sum_{i=1}^U\mathbf{W}_{i}
+ \sum_{t=1}^N\mathbf{V}_{t}\right)\\
& \text{s.\ t.}\ & & \eqref{const11}, \eqref{const13}, \eqref{const12C}, \ \forall i;
\\&&&\eqref{const21},\eqref{const23}, \eqref{const22C},  \ \forall i, \ \forall t;\\&&& \eqref{const31}, \eqref{const33},\eqref{const32C}, \ \forall t,\\
&&&\vartheta_i\geq 0,\ \forall i, \ \ b_t\geq 0,\ \forall t, \\ &&& \vartheta_i^{(t)}\geq 0, \ \forall i, \ \forall t, \\&&&
 \ \mathbf{W}_{i}\succeq \mathbf{0},\ \forall i, \ \mathbf{V}_{t}\succeq \mathbf{0},\ \forall t.
\end{aligned}
\end{equation}
The same technique as presented in Appendix \ref{apen1} can now be adopted to prove that the optimal beamforming matrices $\{\mathbf{W}_{i}\}$ and $\{\mathbf{V}_{t}\}$ in \eqref{prob_sdp_bernstein_Equivalent} are always rank-one if the problem is feasible. In the following, we sketch some of the main steps.

Let $\mathfrak{\overline{L}}\left(\{\mathbf{W}_{i}\},\{\mathbf{V}_{t}\},\overline{\Upsilon}\right)$  be the Lagrangian of \eqref{prob_sdp_bernstein_Equivalent} where $\overline{\Upsilon}$ represents the collection of all Lagrangian multipliers associated with the constraints in \eqref{prob_sdp_bernstein_Equivalent}. Furthermore, let $\overline{g}\left(\overline{\Upsilon} \right)$ be the dual function of \eqref{prob_sdp_bernstein_Equivalent}, i.e., $\overline{g}\left(\overline{\Upsilon} \right)=\displaystyle \min_{\{\mathbf{W}_{i}\},\{\mathbf{V}_{t}\}}
\mathfrak{\overline{L}}\left(\{\mathbf{W}_{i}\},\{\mathbf{V}_{t}\},\overline{\Upsilon}\right)$.
As problem \eqref{prob_sdp_bernstein_Equivalent} is convex and satisfies Slater's constraint qualification \cite{Boyd_convex}, the duality gap is zero and the optimal solution of \eqref{prob_sdp_bernstein_Equivalent} can be obtained by solving its dual problem. Let $\overline{\Upsilon}^{\star}$ be the optimal solution of the dual problem of \eqref{prob_sdp_bernstein_Equivalent}. Then, the corresponding optimal solution $\left(\{\mathbf{W}_{i}^{\star}\},\{\mathbf{V}_t^{\star}\}\right)$ to problem \eqref{prob_sdp_bernstein_Equivalent}
can be obtained as
\begin{eqnarray}\label{dual_funcTheo2}
\overline{g}(\overline{\Upsilon}^{\star})=\displaystyle \min_{\{\mathbf{W}_{i}\},\{\mathbf{V}_{t}\}}
\mathfrak{\overline{L}}\left(\{\mathbf{W}_{i}\},\{\mathbf{V}_{t}\},\overline{\Upsilon}^{\star}\right).
\end{eqnarray}

After some mathematical manipulations, one can rewrite \eqref{dual_funcTheo2} in the following form
\begin{eqnarray}\label{dual_funcNew2}
\overline{g}(\overline{\Upsilon}^{\star})
=\displaystyle \min_{\{\mathbf{W}_{i}\}} \sum_{i=1}^U \textrm{Tr}\left(\mathbf{\overline{\Theta}_{i}} \mathbf{W}_{i}\right)+
\displaystyle \min_{\{\mathbf{V}_{t}\}} \sum_{t=1}^N \textrm{Tr}\left(\mathbf{\overline{\Xi}_{t}} \mathbf{V}_{t}\right),
\end{eqnarray}
where $\overline{\Theta}_{i}$ and $\overline{\Xi}_{t}$ are functions of
$\mathbf{C}\left( \theta_i\right)$, $\mathbf{T}_i\left( \widetilde{\mathbf{h}}_i\right)$, $\mathbf{P}_i\left( \mathbf{H}_i^{1/2}\right)$, $\mathbf{U}_p$, $\mathbf{K}_p$, $\mathbf{C}\left( \theta_i^{(t)}\right)$, $\mathbf{T}_i\left( \widetilde{\mathbf{g}}_t\right)$, $\mathbf{C}\left( a_t\right)$, and $\overline{\Upsilon}^{\star}$. By using a similar contradiction approach as in the proof of Theorem \ref{theo1}, one can show that the optimal solution, $\{\mathbf{W}_{i}^{\star}\},\{\mathbf{V}_t^{\star}\}$, to problem \eqref{dual_funcNew2}, which is also the optimal solution to \eqref{prob_sdp_bernstein_Equivalent}, can only involve rank-one matrices.
\bibliographystyle{IEEEtran}

\end{document}